\documentclass[12pt]{iopart}

\usepackage[utf8]{inputenc}

\usepackage{braket}
\usepackage{cite}
\usepackage{bm}
\usepackage{hyperref}
\usepackage{fancyhdr}

\usepackage{color}
\usepackage[dvipsnames]{xcolor}

\usepackage{amsthm}
\usepackage{amsfonts}
\newtheorem{theorem}{Theorem}
\newtheorem{assumption}{Assumption}
\newtheorem{lemma}{Lemma}

\newtheorem{definition}{Definition}

\newcommand{\E}{\mathcal{E}_{N}}
\newcommand{\linn}{\mathcal{L}_{N}}
\newcommand{\lin}{\mathcal{L}}

\usepackage{graphicx}

\begin{document}

\begin{center}{\large \textbf{{Mean-field dynamics of open quantum systems with collective operator-valued rates: validity and application} }}
\end{center}

\begin{center}
Eliana Fiorelli\textsuperscript{1,2,3*},
Markus M\"uller \textsuperscript{2,3},
Igor Lesanovsky \textsuperscript{4,5},
Federico Carollo \textsuperscript{4}
\end{center}

\begin{center}
{\bf 1} Instituto de F\'isica Interdisciplinar y Sistemas Complejos (IFISC), UIB–CSIC, UIB Campus, Palma de Mallorca, 07122, Spain
\\
{\bf 2} Institute for Theoretical Nanoelectronics (PGI-2), Forschungszentrum J\"{u}lich, 52428 J\"{u}lich, Germany
\\
{\bf 3} Institute for Quantum Information, RWTH Aachen University, 52056 Aachen, Germany
\\
{\bf 4} Institut f\"{u}r Theoretische Physik, Universit\"{a}t T\"{u}bingen, Auf der Morgenstelle 14, 72076 T\"{u}bingen, Germany
\\
{\bf 5} School of Physics and Astronomy and Centre for the Mathematics and Theoretical Physics of Quantum Non-Equilibrium Systems, University of Nottingham, Nottingham, NG7 2RD, UK 
\\
* eliana@ifisc.uib-csic.es 
\end{center}

\begin{center}
\today
\end{center}

\section*{Abstract}
{\bf
We consider a class of open quantum many-body Lindblad dynamics characterized by an all-to-all coupling Hamiltonian and by  dissipation featuring collective ``state-dependent" rates. The latter encodes local incoherent transitions that depend on average properties of the system. This type of open quantum dynamics can be seen as a generalization of classical (mean-field) stochastic Markov dynamics, in which transitions depend on the instantaneous configuration of the system, to the quantum domain. We study the time evolution in the limit of infinitely large systems, and we demonstrate the exactness of the mean-field equations for the dynamics of average operators. We further derive the effective dynamical generator governing the time evolution of (quasi-)local operators. Our results allow for a rigorous and systematic investigation of the impact of quantum effects on paradigmatic classical models, such as quantum generalized Hopfield associative memories or (mean-field) kinetically-constrained models. 
}

\newpage

\section{Introduction}

Open quantum many-body systems constitute a fascinating subject of investigation \cite{BreuerP:2002,Lindblad76}. The interplay between coherent and dissipative processes, combined with the large number of microscopic constituents forming the system, can give rise to interesting nonequilibrium  stationary or dynamical phases \cite{diehl2008,diehl2010,dallatorre2010,Schindler2013,tauber2014,marcuzzi2016,minganti2018,iemini2018,carollo2019,chertkov2022} as well as to nonequilibrium critical dynamics \cite{sieberer2013,chertkov2022,helmrich2020,jo2021,jo2022}. An intriguing aspect of the formalism of open quantum systems is that it allows one to start from a purely classical stochastic dynamics (see, e.g., the reaction-diffusion processes considered in Ref.~\cite{hinrichsen2000}), and to gradually introduce quantum effects --- such as quantum superposition --- and analyze their impact on paradigmatic classical models \cite{marcuzzi2016,jo2021,chertkov2022}. 
In the Markovian regime, open quantum dynamics are described by means of quantum master equations [see Eq.~\eref{Lindblad-generator-Schr} below] with time-independent (Lindblad) generators \cite{Lindblad76,BreuerP:2002}. Despite looking fairly simple, solving these quantum master equations is a daunting task due to the exponential growth, with the number of particles, of the resources needed to describe the quantum state. This often renders both their numerical simulation \cite{weimer2021} and their analytical solution impractical. 

One way to make progress and to achieve a first analytical understanding of the behavior of these quantum systems is that of exploiting a mean-field approach  \cite{BenedikterPS15,MerkliR18,Porta16,Pickl11}, which also proved very useful in equilibrium settings \cite{hepp1973,hioe1973}. Broadly speaking, within this framework one neglects correlations in the system and this allows one to find a reduced set of differential equations providing the time evolution of key system observables. Interestingly, in certain cases such an approach can be shown to become exact in the thermodynamic limit, see, e.g., Refs.~\cite{BenedikterPS15,MerkliR18,Porta16,Pickl11,hepp1973,hioe1973}. For what concerns open quantum dynamics, the exactness of the mean-field approach has been rigorously shown for systems with collective jump operators and with a Hamiltonian featuring an all-to-all interaction between the different subsystems\cite{alicki1983,Benatti2016,BenattiEtAl18} as well as for different versions of spin-boson models \cite{MerkliR18,davies1973,mori2013,CarolloL:PRL:21}. The validity of a mean-field approach in certain open quantum systems has also been investigated numerically \cite{kirton2017,shammah2018,huybrechts2020,wang2021,piccitto2021}. 

In this manuscript, we consider quantum systems composed of a large number of finite-dimensional particles subject to a dissipative Markovian time evolution. In particular, we assume their open quantum dynamics to be characterized by an all-to-all coupling Hamiltonian and by dissipative (stochastic) single-body transitions, whose rates depend on the full many-body state. For these open quantum dynamics, we rigorously demonstrate the validity of the mean-field approach, both for the evolution of system-average properties and for the dynamics of local observables. To give a concrete example, our results apply --- but are not limited --- to quantum generalizations of Hopfield-like associative-memory dynamics \cite{Hopfield:1982,Gayrard92}, which are recently receiving attention also due to the possibility of realizing these systems in current experiments  \cite{MarshEtAl:PhysRevX:21}. Our findings put on rigorous footing existing results on their nonequilibrium behavior \cite{RotondoEtal:2018,Fiorelli:PRA:2019,FiorelliLM22}, justifying the investigation of the impact of the quantum effects on these platforms within a mean-field approach.

Our paper is organized as follows. In Section~\ref{overview}, we give a brief overview of our work explaining, in non-technical terms, the setting as well as our findings. In Section~\ref{sec3}, we introduce the system of interest and its Lindblad generator,  while in Section~\ref{sec4} we derive our main results. In Section~\ref{sec5}, as an application of our findings, we discuss the exactness of the mean-field equations for open quantum Hopfield neural networks. Finally, in the Appendixes, we prove several Lemmata needed to demonstrate our main theorems.

\section{Overview of the paper}
\label{overview}
We provide here an overview whose aim is to introduce the class of open quantum dynamics we will focus on, and to motivate their relevance. For concreteness, we limit the discussion of this Section to a system made by an ensemble of two-level particles. Our results are, however, valid for many-body systems made by $d$-level particles, with arbitrary $d<\infty$. 

\subsection{Dissipation with operator-valued rates}
\label{motivation}
We consider a system made by an ensemble of $N$ classical (Ising) spin-$1/2$ particles. Each particle is thus a two-level system, which can either be found in an {\it excited state} $\ket{\bullet}$ or in a {\it ground state} $\ket{\circ}$ [cf.~Fig.~\ref{Fig1}(a)]. For these particles, the simplest stochastic Markovian dynamics one can imagine is that of independent spin-flips. Namely, each particle can change its state either from the excited state to the ground state, $\ket{\bullet}\to\ket{\circ}$, at a rate $\gamma_\circ$, or from the ground state to the excited state, $\ket{\circ}\to\ket{\bullet}$, at a rate $\gamma_\bullet$, as depicted in Fig.~\ref{Fig1}(a). This is a simple non-interacting ``thermal" time evolution for the $N$-body system and does not show particularly interesting dynamical nor stationary features. A more intricate dynamics can emerge when the rate for the single-particle transitions depends on the configuration of the remainder of the system, see, e.g., example in Fig.~\ref{Fig1}(b). For instance, the rate of flipping into the excited state the $k$th spin could depend on whether particles $k-1$ and $k+1$ are in their excited  or in their ground state [cf.~Fig.~\ref{Fig1}(b)]. This scenario typically occurs when considering relaxation dynamics towards thermal states of classical interacting Hamiltonians, where transition rates depend on the difference in the energy before and after the transition \cite{glauber1963,walter2015}. Another interesting framework in which one finds state-dependent rates, is that of kinetically-constrained models \cite{fredrickson1984,cancrini2008,garrahan2011}, where certain transitions may be  forbidden if a given constraint is not satisfied. For instance, in the example shown in Fig.~\ref{Fig1}(b), we illustrate a model in which a change of the state for a given particle can only occur if both the neighboring particles are excited. In certain cases, e.g., with collective all-to-all classical Hamiltonian functions, transition rates depend on collective properties of the system. A possible generalization of the example in Fig.~\ref{Fig1}(b) to collective rates is achieved by choosing rates to depend on the square of the operator describing the density of excited states in $\ket{\bullet}$, i.e., $n_\bullet=\frac{1}{N}\sum_{k=1}^{N}n^{(k)}$, where $n^{(k)}$ is the operator $n=\ket{\bullet}\!\bra{\bullet}$ for the $k$th particle [see an illustration in Fig.~\ref{Fig1}(c)].

\begin{figure}[t]
\centering
\includegraphics[width=0.8\textwidth]{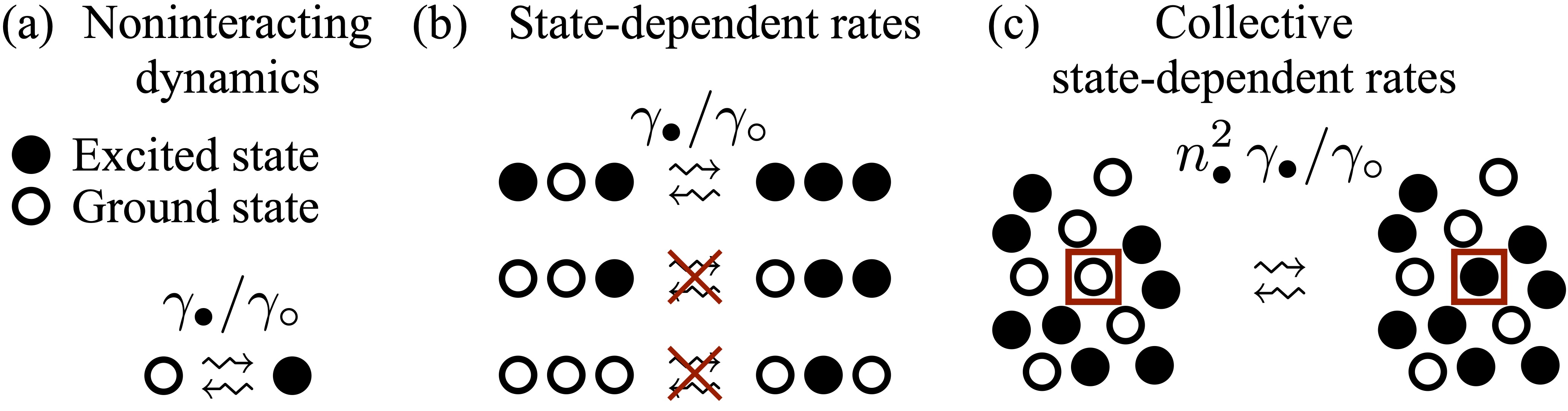}
\caption{{\bf Collective state-dependent rates.} a) A two-level system can either be found in an occupied state  $\bullet$ or in an empty one $\circ$. The simplest classical stochastic non-interacting dynamics for an ensemble of several two-level systems is that of independent spin-flips $\bullet \to\circ$ (rate $\gamma_\circ$) or $\circ\to\bullet$ (rate $\gamma_\bullet$).  In this case rates for the different transitions do not depend on the state of the neighboring particles. b) Example of a kinetically-constrained model in which the central particle can change its state only if the neighboring ones are both in the occupied state. c) In a collective all-to-all model,  the dynamics sketched in panel b) would reduce to one with transition rates which depend on the square of the density of occupied particles $n_\bullet$. }
\label{Fig1}
\end{figure}

This dynamics, just like any classical stochastic dynamics, can be written within the density-matrix formalism of open quantum systems \cite{Garrahan18}. This is done by introducing a dynamical generator --- which preserves diagonal density matrices (see, e.g., Ref.~\cite{CarolloGK:JSP:21}) --- as follows \footnote{\label{gen-dualgen}We denote dynamical generators acting on density-matrices with a $*$, as done for $\mathcal{D}^*$. We use $\mathcal{D}$ to denote instead the generator implementing the time evolution of observables.}
$$
\mathcal{D}^{*}[\rho]=\sum_{k=1}^N \left(J_\bullet^{k}\rho J_\bullet^{k\, \dagger }-\frac{1}{2}\left\{J_\bullet^{k\, \dagger } J_\bullet^{k},\rho \right\}\right)+\sum_{k=1}^N \left(J_\circ^{k}\rho J_\circ^{k\, \dagger }-\frac{1}{2}\left\{J_\circ^{k\, \dagger } J_\circ^{k},\rho \right\}\right)\, ,
$$
with  
$$
J_\bullet^k=\sqrt{\gamma_\bullet} \sigma_+^{(k)} n_\bullet\, , \qquad \qquad J_\circ^k=\sqrt{\gamma_\circ}\sigma_-^{(k)} n_\bullet\, ,
$$
and $\sigma_{+}=\ket{\bullet}\!\bra{\circ}$,  $\sigma_-=\sigma_+^\dagger$. This generator  evolves an initial density matrix $\rho$, through the equation $\dot{\rho}_t=\mathcal{D}^*[\rho_t]$. In the example above, the rates are operator-valued functions of a collective observable, namely the density of excited particles $n_\bullet$.

While formulated in a quantum language, the above dynamics is fully classical (whenever starting from a diagonal state). Nonetheless, it is now straightforward to add quantum coherent Hamiltonian contributions to such a dissipative stochastic time evolution and to investigate their impact on the behavior of the system. This can be done by considering the more general quantum master equation
\begin{equation}
\dot{\rho}_t=\mathcal{L}^*[\rho_t]:=-i[H,\rho_t]+\mathcal{D}^*[\rho_t]\, .
\label{Lindblad-generator-Schr}
\end{equation}
In this paper, we shall consider Lindblad generators with an all-to-all interacting  Hamiltonian $H$, and with dissipation characterized by collective state-dependent rates. 

\subsection{Contribution of this work}
In this work we derive the time evolution of average operators, such as the average ``magnetization" operators for the spin system $m_\alpha^N=\sum_{k=1}^N\sigma_\alpha^{(k)}/N$ (where $\sigma_\alpha$ are Pauli matrices constructed from the basis states $\ket{\bullet},\ket{\circ}$), under the dynamics generated by Lindblad operators of the form discussed in Subsection \ref{motivation} [see also Eq.~(\ref{Lindblad}-\ref{dissipator}) below], in the thermodynamic limit. In particular, we show the validity of the mean-field approximation --- obtained by factorizing expectation values of average operators (see discussion in Section \ref{Heis-MF}) --- for these models. The corresponding proof is presented in Section \ref{exactness} and follows the approach developed in Ref.~\cite{CarolloL:PRL:21}. In Section \ref{generator-quasi-local}, we further derive the effective Lindblad generator implementing the time-evolution of any (quasi-) local operator, such as a single-spin operator  $\sigma_\alpha^{(k)}$, in the thermodynamic limit.

\section{Model systems and their dynamical generators}
\label{sec3}
In this section, we present the class of systems under investigation, and we introduce the algebra of operators as well as a functional representation of the quantum states \cite{BratteliR82}. We then move to the definition of the so-called average operators --- which are nothing but sample-mean averages of a same single-particle operator \cite{BenattiEtAl18,Verbeure10} over the whole system --- and discuss their properties when considering clustering states \cite{LandfordR69,Strocchi05,thirring2013quantum}, i.e., states with sufficiently short-ranged --- in a sense made precise by Definition \ref{clustering} below --- correlations. At the end of the Section, we introduce the general form of the considered dynamical generators and prove first results about their action on local and on average operators.  

\subsection{Quasi-local algebra and quantum states}
We consider a many-body quantum system $S$, consisting of a (countably) infinite number of identical (distinguishable) particles, assumed to be $d$-level systems with $d<\infty$. Each particle can thus be associated with a natural number $k\in\mathbb{N}$. Any single-particle operator $x^{(k)}$, with $x\in M_d(\mathbb{C})$ and $M_d(\mathbb{C})$ being the algebra of $d\times d$ complex matrices, which acts non-trivially only on the $k$th particle can be lifted to be an operator of the many-body system by exploiting a tensor-product structure as 
$$
x^{(k)}={\bf 1}_d\otimes {\bf 1}_d\otimes \dots \otimes x\otimes {\bf 1}_d\otimes {\bf 1}_d \otimes \dots ,
$$
where ${\bf 1}_d$ is the identity in $M_{d}(\mathbb{C})$ and $x$ appears in the $k$th entry of the tensor product. All the (almost local) operators of the many-body system are contained in the so-called {\it quasi-local} $C^*$-algebra $\mathcal{A}$, which is obtained as the norm closure (here and throughout we consider the operator norm, denoted as $\|\cdot\|$, given by the largest eigenvalue, in modulus, of the operator) of the union of all possible local sub-algebras of the system \cite{BratteliR82}. In practice, the quasi-local algebra $\mathcal{A}$ contains all strictly local operators, i.e., all operators supported on a finite number of particles, as well as those operators which are quasi-localised, i.e., they are extended over the whole system but happen to be the limit of a converging sequence of local operators.

The full information about the state of a physical system is equivalent to the knowledge of all possible expectation values for its operators. Thus, given the algebra $\mathcal{A}$, the state of a quantum system can be generically represented as a functional, $\omega$, associating to each operator $A\in\mathcal{A}$ a complex number $\langle A\rangle$ embodying the expectation of the operator itself, $\mathcal{A}\ni A\mapsto \omega(A)=\langle A\rangle$. In order for such a functional to describe a physically-consistent state, $\omega$ must be a linear, positive and normalized [$\omega({\bf 1})=1$ with  ${\bf 1}$ being the identity of $\mathcal{A}$] functional on the quasi-local algebra \cite{BratteliR82}. In certain cases, the expectation values of single-particle operators do not depend on the considered particle, i.e., for any $x\in M_d(\mathbb{C})$ we have $\omega(x^{(k)})=\langle x\rangle$, $\forall k\in \mathbb{N}$. In these cases, the state is called translation invariant [see also Definition \ref{clustering} below]. 

\subsection{Average operators}
The quasi-local algebra $\mathcal{A}$ is the algebra of all operators which are, roughly speaking, almost localised in certain regions of the system. Often, however, when one considers many-body systems it is important to look at the behavior of collective operators, which can account for average properties of the whole system. For instance, this is the case when studying equilibrium as well as nonequilibrium phase transitions, which can be investigated and characterized via the behavior of so-called order-parameters. 

We are interested in the behavior of sequences of operators of the form 
\begin{equation}
X_{N} \equiv \frac{1}{N} \sum_{k=1}^{N} x^{(k)}, \qquad \mbox{with } \qquad x\in M_d(\mathbb{C})\, .
\label{eq:average-operators}
\end{equation}
These operators represent sample-mean averages of a same single-particle operator and are related to the random variables appearing in the law of large numbers \cite{grimmett2020probability}. For each finite $N$, the number of particles considered in the above summation is finite and thus the operator is strictly local. However, we are interested in the behavior of the average operators when $N\to\infty$. 

It turns out that the commutator between any two average operators, $[X_N,Y_N]$, goes to zero in the large $N$ limit, since its norm is bounded by $2\|x\| \|y\|/N$ \cite{LandfordR69,BratteliR82,BenattiEtAl18}. As such, these operators give rise to an emergent classical algebra in the thermodynamic limit. Still, the limiting point $X_\infty$ of the sequence $X_N$ in Eq.~\eref{eq:average-operators} does not belong to the quasi-local algebra $\mathcal{A}$, since the sequence $X_N$ does not converge in the norm topology \cite{BratteliR82}. To understand the structure of these operators in the thermodynamic limit $N\to\infty$, we need to resort to weaker forms of convergence. Here, we consider the so-called {\it weak operator topology} \cite{Strocchi05}. We will say that a sequence of operators $C_n$ converges weakly to the operator $C$, formally denoted as $C=(\mathrm{w\mbox{--}})\!\lim_{n\to\infty} C_n$, (to be read as weak-limit of the sequence $C_n$), if \footnote{We note that this form of convergence coincides with the  weak operator convergence within the so-called GNS representation of the algebra $\mathcal{A}$ induced by the state $\omega$ \cite{BratteliR82,Strocchi05}} 
\begin{equation}
\lim_{n\to\infty} \omega( A^\dagger C_n B)=\omega(A^\dagger  C B)\, \qquad \forall A,B\in \mathcal{A}\, . 
    \label{w-lim}
\end{equation}
This apparently abstract definition has a very relevant physical meaning: in the weak operator topology, we obtain information on the nature of the limiting operator $C$, by controlling all of its possible correlation functions with any quasi-local operator under the expectation associated with the quantum state $\omega$. 

For clustering quantum states, i.e., for states with sufficiently short-ranged correlations, the limiting operators $X_\infty$ of the sequences $X_N$ are nothing but multiples of the identity \cite{LandfordR69,BratteliR82,Strocchi05,Verbeure10}. This means that, $X_\infty={\rm (w\mbox{--}})\!\lim_{N\to\infty}X_N=\langle x\rangle$, where $\langle x\rangle=\omega(x)$ is the expectation of the single-particle operator $x$, where we have further assumed translation invariance of the state. (Note that on the right-hand side of the above limit the complex number $\langle x\rangle$ should be multiplied by an identity operator ${\bf 1}$. However, in order to simplify the notation we omit writing this here and in the following.) This occurs for instance for so-called {\it ergodic} states, i.e., for states that obey 
$$
\omega(x^{(k)}y^{(h)})\approx \omega(x^{(k)})\omega (y^{(h)})
$$
whenever $|k-h|$ is sufficiently large (see more general definition in, e.g., Ref.~\cite{Verbeure10}). Since in our work we will mainly look at average operators, we define clustering states through the property highlighted in Eq.~\eref{def_clust} of the following Definition.

\begin{definition}\label{clustering}
We refer to quantum states $\omega$ of the quasi-local algebra $\mathcal{A}$ as translation-invariant clustering states if the following properties are satisfied: 
\begin{eqnarray}
    & \mbox{i)} \quad \omega(x^{(k)})=\omega(x^{(h)})=\langle x\rangle, \qquad &\forall x\in M_d(\mathbb{C}), \forall k,h\in \mathbb{N}\, ; \\
    & \mbox{ii)}\lim_{N\to\infty}\omega([X_N-\langle x\rangle]^2)=0, \qquad &\forall x=x^\dagger \in M_d(\mathbb{C})\, . \label{def_clust} 
\end{eqnarray}
\end{definition}
The second property above shows that for such clustering states the variance of the operators $X_N$ vanishes in the large $N$ limit and, thus, the limiting operators $X_\infty$ must converge to multiples of the identity. It is indeed possible to show that Eq.~\eref{def_clust} implies the weak convergence of $X_N$ to $X_\infty=\langle x\rangle$, as defined by Eq.~\eref{w-lim}.

\subsection{Lindblad generators with collective operator-valued rates}
We assume that the many-body system introduced above is subject to a Markovian open quantum dynamics \cite{Lindblad76,BreuerP:2002}, implemented through a quantum master equation by means of a time-independent dynamical generator. The latter must assume a Lindblad form for the dynamics to be physically consistent \cite{Lindblad76}. 

The time-evolution of any operator $O\in \mathcal{A}$ must thus obey the equation
\begin{equation}
\dot{O}(t)=\mathcal{L}_N[O(t)]\, ,
    \label{QME}
\end{equation}
with $\mathcal{L}_N$ being the Lindblad operator evolving observables, i.e., the generator dual to the one introduced in Eq.~\eref{Lindblad-generator-Schr} (see also Footnote \ref{gen-dualgen}). The formal solution of the above equation is given by $O(t)=e^{t\mathcal{L}_N}[O]$. As usually done in order to study the emergent dynamics in the infinite system, we have first defined the dynamical generator $\mathcal{L}_N$ for an ensemble of $N$ particles, and we will then derive the asymptotic dynamics taking the limit $N\to\infty$. 

Before discussing the form of the considered dynamical generator (briefly mentioned in Section \ref{motivation}), it is convenient to introduce an orthonormal, hermitian basis $\{v_\alpha\}_{\alpha=1}^{d^2}$ for the single-particle algebra $M_{d}(\mathbb{C})$. We thus have a set of operators such that $v_\alpha=v_\alpha^\dagger$ as well as $\tr{(v_{\alpha} v_{\beta})}=\delta_{\alpha \beta}$ (implying  $\|v_\alpha\|\le 1$)  which we can employ to decompose any other operator $x\in M_d(\mathbb{C})$ through the relation
\begin{equation}
x=\sum_{\alpha=1}^{d^2} \tr{(x \, v_{\alpha})}v_{\alpha}\, .
\end{equation}
For later convenience, we also define the structure coefficients $a_{\alpha\beta}^{\gamma}$ for the chosen basis, obtained as 
\begin{equation}
[v_{\alpha},v_{\beta}] = \sum_{\gamma=1}^{d^2}  a_{\alpha\beta}^{\gamma}  v_{\gamma}, \quad a_{\alpha\beta}^{\gamma} \equiv \tr{([v_{\alpha},v_{\beta}] v_{\gamma})}.
\end{equation}

Exploiting this single-particle basis, the Lindblad generator can be decomposed into two different contributions 
\begin{equation}
 \mathcal{L}_N[O]=i[H,O]+\sum_{\ell=1}^{q} \mathcal{D}_\ell[O]\, ,
    \label{Lindblad}
\end{equation}
where $H$ is the Hamiltonian of the system assuming the form 
\begin{equation}\label{e0_totalHamiltonian}
H= \sum_{k=1}^{N} \sum_{\alpha=1}^{d^2} \epsilon_{\alpha} v_{\alpha}^{(k)} + \frac{1}{N} \sum_{k,j=1}^{N} \sum_{\alpha, \beta=1}^{d^2}  h_{\alpha \beta} v_{\alpha}^{(k)} v_{\beta}^{(j)}\, .
\end{equation}
The first term on the right-hand side of the above equation (with $\epsilon_\alpha$ real) represents a single-particle contribution to the Hamiltonian, while the second one, with $h_{\alpha\beta}=h_{\beta\alpha}^*$ considers two-body interactions in an all-to-all fashion. We note that since we have an unconstrained sum --- which double counts the interactions between particles --- the terms $h_{\alpha\beta}$ are equal to  half of the actual interaction strength. Moreover, the double sum also contains terms with $k=j$ which describe single-particle terms rather than interactions. Due to the presence of the factor $1/N$ in front of the second part of the Hamiltonian, these terms become irrelevant in the thermodynamic limit. We can thus safely keep them as this will be convenient later on. In summary, the second contribution to the Hamiltonian in Eq.~\eref{e0_totalHamiltonian} describes interactions between all pairs of particles  with a same strength proportional to $1/N$.

The terms collected in the maps, or dissipators,  $\mathcal{D}_\ell$ describe instead dissipative contributions to the time-evolution. As already discussed in Section \ref{motivation}, we take them to be of the form 
\begin{equation}
    \mathcal{D}_\ell[O]=\frac{1}{2}\sum_{k=1}^N \left(\left[J_{\ell}^{k\,  \dagger}, O\right] J_{\ell}^{k} + J_{\ell}^{k\, \dagger }[O, J_{\ell}^{k}]  \right) \, ,
    \label{dissipator}
\end{equation}
with
\begin{equation}
J_\ell^{k}=j_\ell^{(k)}\Gamma_\ell(\Delta_N^\ell)
    \label{jumps}
\end{equation} 
being the jump operators. Here, $j_\ell^{(k)}$ acts solely on site $k$ while $\Gamma_\ell(\Delta_N^\ell)=[\Gamma_\ell(\Delta_N^\ell)]^\dagger$ is an operator-valued function computed for the operator $\Delta_N^\ell=[\Delta_N^\ell]^\dagger$. We assume the latter operator to be a linear combination with real coefficients of average operators of the type defined in Eq.~\eref{eq:average-operators}, i.e., 
\begin{equation}
\Delta_N^\ell=\sum_{\alpha=1}^{d^2}r_{\ell\alpha} \left[\frac{1}{N}\sum_{k=1}^{N}v_\alpha^{(k)}\right]\,, \qquad \mbox{ with }r_{\ell \alpha}\in \mathbb{R}.
\label{delta}
\end{equation}
From their definition, we see that these  operators are bounded in norm, i.e., $\|\Delta_N^\ell\|\le \delta_\ell$, where 
\begin{equation}
 \delta_\ell =\sum_{\alpha=1}^{d^2}|r_{\ell\alpha}|<\infty\, .
 \label{delta_alpha}
\end{equation}

As discussed in Section \ref{motivation}, the structure of the jump operators $J^k_\ell$ suggests that the function $\Gamma_\ell(\Delta_N^\ell)$, when squared, gives rise to an operator-valued rate for the transition implemented by $j^{(k)}_\ell$ on the $k$th particle. Since the operator $\Delta_N^\ell$, which is the argument of the function, is an average operator, the rate has the structure of a mean-field rate which accounts for an average (collective) property of the system. We consider functions $\Gamma_\ell(\Delta_N^\ell)$ satisfying the following Assumption.

\begin{assumption}
\label{Gamma}
The operator-valued functions $\Gamma_\ell(\Delta_N^\ell)$ can be written as power series 
$$
\Gamma_\ell(\Delta_N^\ell)=\sum_{n=0}^\infty c_\ell^n (\Delta_N^\ell)^{n}\, , 
$$
with coefficient $c_\ell^n$ such that for any $z\in\mathbb{R}$
\begin{equation}
    \gamma(z)=\sum_{n=0}^\infty |c_\ell^n||z|^n <\infty\, .
\end{equation}
For later convenience, we note that the assumption on the series $\gamma(z)$ also implies that
$$
\gamma'(z):=\sum_{n=0}^\infty n|c_\ell^n||z|^{n-1} <\infty\, .
$$
\end{assumption}
The above assumption specifies that we are considering functions $\Gamma_\ell$ which admit a Taylor expansion, around zero, with infinite radius of convergence. This is a strong assumption, which we make here for the sake of simplicity. It is not strictly necessary to prove our main theorems. In Section \ref{sec5}, we show indeed that our approach can be applied also for certain operator-valued rates which do not obey Assumption \ref{Gamma}. Working within this assumption allows us to find results for a broad class of dynamical generators. 

Considering Assumption \ref{Gamma}, we can now readily prove the following result.

\begin{lemma}
\label{lemma_commutators}
For any given operator-valued function $\Gamma_\ell(\Delta_N^\ell)$ satisfying Assumption \ref{Gamma}, the following relations hold 
\begin{eqnarray*}
	&i) \left\| \left[\Gamma_\ell(\Delta_N^\ell), O\right] \right\| \le \frac{2N_O}{N}\|O\|\delta_\ell \gamma'(\delta_\ell) \, ,\\
	&ii) \left\| \left[\Gamma_\ell(\Delta_N^\ell), X_{N}\right] \right\| \le \frac{2}{N}\|x\| \delta_\ell  \gamma'(\delta_\ell)\, ,\\
	&iii) \left\| \left[\Gamma_\ell(\Delta_N^\ell),\left[\Gamma_\ell(\Delta_N^\ell), O\right]\right] \right\| \le \frac{4N_O^2}{N^2}\|O\|\delta_\ell^2 [\gamma'(\delta_\ell)]^2  \, , \\
	& iv)  \left\| \left[\Gamma_\ell(\Delta_N^\ell),\left[\Gamma_\ell(\Delta_N^\ell), X_N\right]\right] \right\| \le \frac{4}{N^2}\|x\|\delta_\ell^2 [\gamma'(\delta_\ell)]^2    \, ,
\end{eqnarray*}
with $O$ being any operator with strictly local support, $N_O$ the length of such support, and $X_N$ any average operator as defined in Eq.~\eref{eq:average-operators}.
\end{lemma}

The proof of the above Lemma is shown in \ref{app_lemma_comm} and simply requires the evaluation of the commutators between the operator-valued rates and local operators, by exploiting the power series expansion of the former. 

We conclude this Section stating a first result concerning the action of the considered dynamical generators on strictly local observables and on average operators. In particular, we show that, in the thermodynamic limit $N\to\infty$, the maps $\mathcal{D}_\ell$ act on these operators as if they were local maps weighted by a pre-factor equal to $\Gamma_{\ell}^2(\Delta_N^\ell)$. This  corroborates the intuition that the considered jump operators implement local transitions associated with operator-valued rates. Formally, this is expressed by the following Lemma:

\begin{lemma}
\label{Cor_diss_dyn} 
The maps $\mathcal{D}_\ell$ defined by Eqs.~\eref{dissipator}-\eref{jumps} with functions $\Gamma_\ell(\Delta_N^\ell)$ obeying Assumption \ref{Gamma} are such that 
\begin{eqnarray*}
	& \left\|\mathcal{D}_\ell[O]-\Gamma_{\ell}^2(\Delta_N^\ell) \mathcal{D}_\ell^{\rm Loc}[O]\right\|\le \frac{C_O}{N}\, ,\\
	& \left\|\mathcal{D}_\ell[X_{N}]-\Gamma_{\ell}^2(\Delta_N^\ell) \mathcal{D}_\ell^{\rm Loc}[X_N]\right\|\le \frac{C_{x}}{N}\, , \qquad 
\end{eqnarray*}
with 
\begin{equation}
\mathcal{D}^{\rm Loc}_\ell[A]=\frac{1}{2}\sum_{k=1}^N \left(\left[{j}_{\ell}^{ \dagger\, (k)}, A\right] {j}_{\ell}^{(k)} + {j}_{\ell}^{ \dagger\, (k)}[A, {j}_{\ell}^{(k)}]  \right)\, ,
\label{D_loc}
\end{equation}
and 
$C_O$, $C_{x}$ appropriate $N$-independent constants. In the above expression, $O$ is any local operator with support on a finite number of sites, $N_O$ is the extension of its support, and $X_N$ the average operator constructed from the  single-particle operator $x$, as shown in Eq.~\eref{eq:average-operators}.
\end{lemma}
The proof of the Lemma, which is reported in \ref{app_corol_eqD}, together with the expression of the constant $C_O, C_{x}$, requires the direct evaluation of the action of the dissipator $\mathcal{D}_\ell $ on the operators $O$, $X_N$.

\section{Main Results}
\label{sec4}
In this section, we present the main results of our paper. In the first subsection, we derive the Heisenberg equations of motion for average operators of single-particle observables and present the mean-field equations. The latter, as we explain, are obtained by factorizing expectation values of average operators. In the second subsection, we show that these mean-field equations are in fact exactly reproducing the time evolution of the considered average operators, in the thermodynamic limit. In the last section, we exploit the result obtained for the average operators to derive an effective dynamical map for any quasi-local operator of the many-body system. 

\subsection{Heisenberg equations and mean-field dynamics}
\label{Heis-MF}
As already anticipated in the previous section, in the study of many-body systems one is often interested in understanding the dynamical or the stationary behavior of collective operators. Here, in particular, we are interested in the average operators of Eq.~\eref{eq:average-operators}. Given that we defined a basis for the single-particle algebra we can construct a basis for all possible average operators. We define the following set of average operators
\begin{equation}
    m_\alpha^N=\frac{1}{N}\sum_{k=1}^N v_\alpha^{(k)}\, , \qquad \alpha=1,2,\dots d^2\, ,
\end{equation}
and note that any $X_N$ can be obtained as a linear combination of the above operators. The goal is thus to describe the time evolution of $m_\alpha^N$, $e^{t\mathcal{L}_N}[m_\alpha^N]$, in the thermodynamic limit, starting from a translation invariant clustering state as in Definition \ref{clustering}. The first step in this direction is to compute the Heisenberg equation of motion, namely 
\begin{equation}
    \frac{d}{dt} e^{t\mathcal{L}_N}[m_\alpha^N]=\mathcal{L}_N\left[e^{t\mathcal{L}_N}[m_\alpha^N]\right]=e^{t\mathcal{L}_N}\left[\mathcal{L}_N [m_\alpha^N]\right]\, .
    \label{eq:Heisenberg}
\end{equation}
To this end, one needs to control the action of the Lindblad operator on average operators, as shown by the second equality in the above equations. This is done with the following Lemma. 

\begin{lemma}
\label{lemma_gen_action}
Given the generator $\lin_N$ specified by Eqs.~\eref{Lindblad}-\eref{jumps}, with functions $\Gamma_\ell(\Delta_N^\ell)$ obeying Assumption \ref{Gamma}, we have that 
$$
\| \mathcal{L}_N[m^{N}_{\alpha}] -  f_\alpha(\vec{m}^N)\| \leq \frac{C_{L}}{N}
$$
where 
\begin{eqnarray*}
& f_\alpha(\vec{m}^N) =  i \sum_{\beta = 1}^{d^2} A_{\alpha \beta} m^{N}_{\beta} + i \sum_{\beta,\gamma=1}^{d^2} B_{\alpha \beta \gamma} m^{N}_{\beta} m^{N}_{\gamma} + \sum_{\ell, \beta} M_{\ell \alpha }^{\beta} \Gamma^{2}_{\ell}(\Delta^{\ell}_N) m^{N}_{\beta} \\
& A_{\alpha \beta} = \sum_{\beta'=1}^{d^2} \epsilon_{\beta'} a_{\beta' \alpha}^{\beta} \quad B_{\alpha \beta \gamma} = \sum_{\beta'=1}^{d^2} a_{\beta' \alpha}^{\gamma}(h_{\beta \beta'} + h_{\beta' \beta})\, .
\end{eqnarray*} 
Here, $M$ is a real matrix, such that the action of $\mathcal{D}^{\mathrm{Loc}}_{\ell}[\cdot]$ on an element of the single-site operator basis $v_\alpha^{(k)}$ reads
$$
\mathcal{D}_\ell^{\rm Loc}[v_\alpha^{(k)}]=\sum_{\beta=1}^{d^2}M_{\ell \alpha}^{\beta} v_\beta^{(k)}\, , 
$$
and $C_{L}$ is an $N$-independent bounded quantity.
\end{lemma}

To prove this Lemma and deriving the constant $C_L$, as shown in  \ref{app_proof_L3}, one needs to evaluate the generator $\linn$ on $m_{\alpha}^N$ and  exploit Lemma \ref{Cor_diss_dyn}. 

It is worth stressing that Lemma \ref{lemma_gen_action} shows that, in the thermodynamic limit, the action of the Lindblad generator on the average operators $m_{\alpha}^N$ can be written as a nonlinear function $f_\alpha$ of the average operators themselves, for any $\alpha$. However, one cannot yet solve the emergent system of equations. The point is indeed that open quantum dynamics are not represented by an automorphism and thus, in principle, $e^{t\mathcal{L}_N}[m_\alpha^N m_\beta^N]\neq e^{t\mathcal{L}_N}[m_\alpha^N]e^{t\mathcal{L}_N}[m_\beta^N]$. As such, due to the nonlinear function $f_\alpha$ of the average operators, the Heisenberg equations of Eq.~\eref{eq:Heisenberg} are not closed on the operators $e^{t\mathcal{L}_N}[m_\alpha^N]$. To proceed one needs to compute the action of the generator on the function $f_\alpha$. This gives rise to an infinite, in the thermodynamic limit, hierarchy of equations which can rarely be solved. 

The evolution equations for the expectation values of the average operators, which are ultimately what one is interested in, can be derived from Eq.~\eref{eq:Heisenberg} by taking the expectation value on both sides of the equations. By defining $\omega_t\left(A\right):=\omega\left(e^{t\mathcal{L}_N}[A]\right)$ and recalling Lemma \ref{lemma_gen_action}, we have 
\begin{equation}
    \frac{d}{dt} \omega_t\left(m_\alpha^N\right)\approx \omega_t\left(f_\alpha(\vec{m}^N)\right)\, .
    \label{eq:expect}
\end{equation}
Clearly, by taking the expectation one cannot make much progress in solving the system. However, at this level, it is straightforward to introduce the so-called mean-field equations of motion. These are obtained by assuming that one can factorize the expectation of products of average operators into the product of the expectations, e.g.,  assuming that 
\begin{equation}
    \omega_t(m^N_\alpha m^N_\beta)\approx \omega_t(m^N_\alpha)\omega_t(m^N_\beta)\, .
    \label{approx}
\end{equation}
Applied to Eq.~\eref{eq:expect}, this leads to the mean-field equations of motion for the considered Lindblad generator given by 
\begin{equation}
    \frac{d}{dt} m_\alpha= f_\alpha(\vec{m})\, .
    \label{eq:mean-field}
\end{equation}
Here $m_\alpha$ is a function of time which should capture the dynamics of $\omega_t\left(m_\alpha^N\right)$ in the thermodynamic limit. The above system of nonlinear ordinary differential equations can in principle be solved analytically or in any case simulated efficiently, taking as initial conditions the ones associated with the initial state $\omega$
\begin{equation}
    m_\alpha(0)=\lim_{N\to\infty}\omega (m_\alpha^N)\, .
    \label{init-cond}
\end{equation}
However, it still remains to be shown that these equations provide the exact dynamics of average operators. Intuitively, for this to be the case, one would expect that the time-evolved state $\omega_t$ is, in the thermodynamic limit, a clustering state (in the sense of Definition \ref{clustering}), so that average operators converge to multiples of the identity justifying the approximation in Eq.~\eref{approx}.

In the next section we prove that this is indeed the case. Before going to that, we state here a technical result on the system of equations in Eq.~\eref{eq:mean-field}, which is proved in \ref{app_proof_L4}.

\begin{lemma}\label{lemma_bound_mf}
The system of equations \eref{eq:mean-field} with initial conditions $m_\alpha(0)$, defined by a quantum state $\omega$ as in Eq.~\eref{init-cond}, has a unique solution for $t\in[0,\infty)$. Morever, one has 
$$
|m_\alpha(t)|\le \|v_\alpha\|\le 1\, , \qquad \forall t\in[0,\infty)\, .
$$
\end{lemma}

\subsection{Exactness of mean-field equations for  average operators}
\label{exactness}
In order to show the exactness of the mean-field equations of motion, one has to show that 
\begin{equation}\label{mf_limit}
\lim_{N \rightarrow \infty }\omega_t(m_{\alpha}^{N}) - m_{\alpha}(t) = 0, \qquad \forall t.
\end{equation} 
As briefly mentioned at the end of the previous section, in order to prove the validity of the above limit, we want to show that the state $\omega_t$ is clustering, as in Definition \ref{clustering}, when considering all relevant average operators. We do this by exploiting the approach discussed in Ref.~\cite{CarolloL:PRL:21}. We define the quantity 
\begin{equation}
\E (t) = \sum_{\alpha = 1}^{d^2} \omega_t ([m_{\alpha}^{N} -m_{\alpha}]^{2}) \,,
\end{equation}
which is a sum of positive contributions and is thus zero only when all terms vanish. Moreover, each term consists of the expectation value of the square of the distance of the operators from their mean-field counterpart. Each of the summands thus considers how close to a multiple of the identity (given by the mean-field operators) the average operators are. Namely, if $\lim_{N\to\infty}\E (t)=0$, then the state $\omega_t$ is clustering in the thermodynamic limit, and this can be used to show that the mean-field equations are exact. Indeed, via the Cauchy-Schwarz inequality, we have $|\omega_t(m_{\alpha}^{N}-m_{\alpha})| \le \sqrt{ \omega_t([m_{\alpha}^{N}-m_{\alpha}]^2)} \le \sqrt{\mathcal{E}_N(t)}$, which can thus be used to control the limit in Eq.~\eref{mf_limit}. Before going ahead we state here a Lemma, which will be useful for the proof of the exactness of the mean-field equations.

\begin{lemma}
\label{Lemma-aux}
The convergence of the square of the operator-valued function $\Gamma_\ell^2(\Delta_N^\ell)$ to the same function computed through the mean-field operators $\Gamma_{\ell}^2(\Delta_{\ell}(t))$, with 
\begin{equation}\label{delta_mean_field}
\Delta_\ell(t)=\sum_\alpha r_{\ell\alpha} m_\alpha(t)\, , 
\end{equation} 
is dominated by the convergence of the average operators $m_\alpha^N$ to the mean-field variables $m_\alpha(t)$ as
\begin{eqnarray*}
& \left|\omega\left(A^\dagger e^{t\mathcal{L}_N}\left[\left(\Gamma_\ell^2(\Delta_N^\ell)-\Gamma_\ell^2(\Delta_\ell(t))\right)X\right] B\right)\right| \\
& \le  2\gamma(\delta_{\ell})\gamma'(\delta_\ell) \| X \| \sum_{\alpha=1}^{d^2} |r_{\ell \alpha}| \sqrt{\omega(A^{\dagger} e^{t \linn}[(m_{\alpha}^N-m_{\alpha}(t))^2]A)}\sqrt{\omega(B^{\dagger}B)} \, .
\end{eqnarray*}
Here, $A,B,X$ can be either quasi-local operators or functions of average operators. 
\end{lemma}

A proof of the above bound to the convergence of the square of the function $\Gamma_\ell(\Delta_N^\ell)$ can be found in  \ref{proof_Lemma-aux}, and exploits a generalized Cauchy-Schwarz inequality proved in Lemma \ref{Lemma-dilation} (see \ref{app_lemma6}). 

With this result, we are now ready to state the first main theorem of our paper, establishing the exactness of the mean-field equations for the time-evolution of average operators.

\begin{theorem}\label{theorem}
Given a generator as the one in Eqs.~\eref{Lindblad}-\eref{delta}, with functions $\Gamma_{\ell}(\Delta^{\ell}_N)$ satisfying Assumption \ref{Gamma}, we have that 
\begin{equation}\label{e_theorem}
if \quad \lim_{N \rightarrow \infty} \E(0) = 0, \quad then \quad  \lim_{N \rightarrow \infty} \E(t) = 0,  \forall  t< \infty .
\end{equation}
implying that the mean-field Eqs.~\eref{mf_limit} are exact. 
\end{theorem}

\begin{proof} The proof of the theorem exploits Gronwall's Lemma. This Lemma states that, if two positive, bounded, $N$-independent constants $C_1$ and $C_2$, such that $\dot{\mathcal{E}}_{N}(t) \leq C_1 \E (t) + C_2/N$ exist, then 
$$
\E(t) \leq e^{C_1 t} \E(0) + C_2(e^{C_1 t}-1)/(C_1 N)
$$ 
which would be enough to prove the theorem using the assumption on the initial value of $\E(0)$ and taking the limit $N\to\infty$. The goal is thus to find such constants $C_1,C_2$.

Let us then inspect the time derivative of $\E(t)$, 
\begin{equation}
\dot{\mathcal{E}}_{N}(t) = \sum_{\alpha=1}^{d^2} \frac{d}{dt} \omega_t \left( [m_{\alpha}^{N} - m_{\alpha}(t)]^2 \right) . 
\end{equation} 
It is convenient to focus on each $\alpha$-th contribution of the latter expression separately, 
\begin{eqnarray}\label{derivative_Dt}
D_t^\alpha =& \frac{d}{dt} \omega_t \left( [m_{\alpha}^{N} - m_{\alpha}(t)]^2 \right)  \\
=& \omega_t(\linn[(m_{\alpha}^N-m_{\alpha}(t))^2])-2 \dot{m}_{\alpha}(t) \omega_t (m_{\alpha}^N-m_{\alpha}(t)),
\end{eqnarray}
where we have exploited that $\dot{\omega}_t(X) = \omega_t(\linn[X])$, and that $m_{\alpha}(t)$, $\dot{m}_\alpha(t)$ are scalar functions. We now focus on the term $\linn[(m_{\alpha}^N-m_{\alpha}(t))^2]$. To this end, it is worth noticing that, given any two operators, $A$, $B$, one has 
\begin{equation}
\linn[AB] = \linn[A]B + A\linn[B] +\sum_{\ell,k} [J_{\ell}^{k \, \dagger},A][B,J_{\ell}^{k}],
\end{equation}
where $J_{\ell}^{k}$ are the jump operators defined by Eq.~\eref{jumps}. Let us focus on the last term, $P \equiv \sum_{\ell,k} [J_{\ell}^{k \, \dagger},m_{\alpha}^N][m_{\alpha}^N,J_{\ell}^{k}]$, where we set $A= B = m_{\alpha}^N-m_{\alpha}(t)$. It is 
\begin{equation}
\eqalign{
P= \sum_{\ell=1}^{q}\sum_{k=1}^{N} & \left(j_{\ell}^{(k) \, \dagger}[m_{\alpha}^N,\Gamma_\ell(\Delta_N^{\ell})]+[m_{\alpha}^N, j_{\ell}^{(k) \, \dagger}]\Gamma_\ell(\Delta_N^{\ell})\right) \\
& \times \left([\Gamma_\ell(\Delta_N^{\ell}),m_\alpha^N]j_{\ell}^{(k)}+\Gamma_\ell(\Delta_N^{\ell})[j_{\ell}^{(k)},m_{\alpha}^N]\right) .
}
\end{equation}
By exploiting Lemma \ref{lemma_commutators}, we have that $\|P \| \leq C_{P} / N$, with $C_{P}$ being the  $N$-independent constant
\begin{equation}
    C_P = \sum_{\ell=1 }^{q}\| j_\ell\|^2 [2 \delta_\ell \gamma'(\delta_\ell) + d^2 a_{\mathrm{max}} \gamma(\delta_\ell)]^2 \, .
\end{equation}
Here, $a_{{\mathrm {max}}}={\mathrm {max}}_{\alpha,\beta,\gamma} |a_{\alpha\beta}^\gamma|$.  Inserting
\begin{equation}
\linn[( m_{\alpha}^N-m_{\alpha}(t) )^2] = \linn[m_{\alpha}^N][m_{\alpha}^N-m_{\alpha}(t)] + [m_{\alpha}^N-m_{\alpha}(t)]\linn[m_{\alpha}^N] +P \,,
\end{equation}
in the time derivative \eref{derivative_Dt}, this reads
\begin{equation}
\eqalign{
D_t^\alpha = & \omega_t([\linn[m_{\alpha}^N]-\dot{m}_{\alpha}(t)][m_{\alpha}^N-m_{\alpha}(t) ])\\
& +\omega_t([m_{\alpha}^N-m_{\alpha}(t) ][\linn[m_{\alpha}^N]-\dot{m}_{\alpha}(t)])  +\omega_t(P) \, .
}
\end{equation}
As the second term on the right-hand side of the above expression is the complex conjugate of the first one, we can focus on the latter,
\begin{equation}\label{partiI_derivative_time}
D_t^{\alpha,I} = \omega_t([\linn[m_{\alpha}^N]-\dot{m}_{\alpha}(t)][m_{\alpha}^N-m_{\alpha}(t) ]) \, .
\end{equation}

Making use of Lemma \ref{lemma_gen_action} and of the mean-field Eqs.~\eref{eq:mean-field} for the quantities $\linn[m_{\alpha}^N]$ and $\dot{m}_{\alpha}(t)$, respectively, we get 
\begin{equation}
\eqalign{
& \linn[m_{\alpha}^N]-\dot{m}_{\alpha}(t) =L+ f_\alpha(\vec{m}^{N})-f_\alpha(\vec{m}(t))   \\ 
& = L+i \sum_{\gamma=1}^{d^2} A_{\alpha \beta} (m_{\beta}^{ N} - m_{\beta}(t)) + i\sum_{\beta, \gamma=1}^{d^2}   B_{\alpha \beta \gamma} (m_{\beta}^{N}  m_{\gamma}^{N}- m_{\beta}(t) m_{\gamma}(t) ) \\
& + \sum_{\ell=1}^{q}  \sum_{\beta=1}^{d^2}\ M_{\ell \alpha}^{\beta} [\Gamma^{2}_{\ell}(\Delta_{N}^{\ell}) m_{ \beta}^{N}- \Gamma^{2}_{\ell}(\Delta_{\ell}(t)) m_{ \beta}(t)] .
\label{aux_Theo_1}
}
\end{equation}
In the above equation $L$ is the operator difference $L=\mathcal{L}_N[m_\alpha^N]-f_\alpha(\vec{m}^N)$, which, as a consequence of Lemma \ref{lemma_gen_action}, obeys $\|L\|\le C_{L}/N$, with $C_{L}$ the  $N$-independent bounded constant obtained in the proof of Lemma \ref{lemma_gen_action}. 

In the second line of the above equation, the last term can be reshaped as follows
\begin{equation}
\eqalign{
 m_{\beta}^{N } m_{\gamma}^{N}- m_{\beta}(t) m_{\gamma}(t)  =  ( m_{\beta}^{N}-m_{\beta}(t))  m_{\gamma}^{N} +  m_{\beta}(t) [m_{\gamma}^{N}(t) - m_{\gamma}(t)],
}
\end{equation}
and, similarly, the last line of Eq.~\eref{aux_Theo_1} can be re-written, adding and subtracting the term $\Gamma^{2}_{\ell}(\Delta_{\ell}(t)) m_{ \beta}^{N}$, as follows 
\begin{equation}
\eqalign{
& \Gamma^{2}_{\ell}(\Delta_N^{\ell}) m_{ \beta}^{N}- \Gamma^{2}_{\ell}(\Delta_{\ell}(t)) m_{ \beta}(t)  = \\  
& \Gamma^{2}_{\ell}(\Delta_{\ell}(t))\left( m_{\beta}^{N}(t)-m_{\beta}(t)\right) +\left[ \Gamma^{2}_{\ell}(\Delta_N^{\ell})-\Gamma^{2}_{\ell}(\Delta_{\ell}(t)) \right] m_{\beta}^{N} \, .
}
\end{equation}
Thus, the expression for $\linn[m_{\alpha}^N]-\dot{m}_{\alpha}(t)$ reads
\begin{equation}
\eqalign{
 \linn[m_{\alpha}^N]-\dot{m}_{\alpha}(t)& =L+ i \sum_{\beta=1}^{d^2} A_{\alpha \beta} (m_{\beta}^{ N} - m_{\beta}(t)) \\
& + i\sum_{\beta, \gamma=1}^{d^2}   B_{\alpha \beta \gamma} [ ( m_{\beta}^{N}-m_{\beta}(t))  m_{\gamma}^{N} +  m_{\beta}(t) (m_{\gamma}^{N}(t) - m_{\gamma}(t))] \\
& + \sum_{\ell=1}^{q}\sum_{\beta=1}^{d^2} M_{\ell \alpha}^{\beta} \left\lbrace \Gamma^{2}_{\ell}(\Delta_{\ell}(t))\left[ m_{\beta}^{N}(t)-m_{\beta}(t)\right]\right. \\
& \left. +  \left[ \Gamma^{2}_{\ell}(\Delta^{\ell}_{N})-\Gamma^{2}_{\ell}(\Delta_{\ell}(t)) \right] m_{\beta}^{N} \right\rbrace \, .
}
\end{equation}
Inserting this expression in the time derivative in Eq.~\eref{partiI_derivative_time}, we get
\begin{equation}
\eqalign{
D_{t}^{\alpha,I} &=\omega_t\left(L[m_\alpha^N-m_\alpha(t)]\right) + i\sum_{\beta=1}^{d^2} A_{\alpha \beta} \, \omega_t \left( [m_{\beta}^{N} -m_{\beta}(t)][m_{\alpha}^{N} -m_{\alpha}(t)] \right) \\
& + i\sum_{\gamma, \beta =1}^{d^2} B_{\alpha \beta \gamma} \left\lbrace \omega_t \left( [m_{\beta}^{N} -m_{\beta}(t)]m_{\gamma}^{N}[m_{\alpha}^{N} -m_{\alpha}(t)] \right) \right. \\
& \left. + \omega_t \left( m_{\beta}(t)[m_{\gamma}^{N} -m_{\gamma}(t)][m_{\alpha}^{N} -m_{\alpha}(t)] \right) \right\rbrace \\
& + \sum_{\ell=1}^{q}\sum_{ \beta=1}^{d^2} M_{\ell \alpha}^{\beta} \left\lbrace \omega_t \left( \Gamma_{\ell}^{2}(\Delta_{\ell}(t))[m_{\beta}^{N} -m_{\beta}(t)][m_{\alpha}^{N} -m_{\alpha}(t)]  \right)  \right. \\ 
& + \left. \omega_t \left(   \left[ \Gamma^{2}_{\ell}(\Delta^{\ell}_{N})-\Gamma^{2}_{\ell}(\Delta_{\ell}(t)) \right] m_{\beta}^{N} [m_{\alpha}^{N} -m_{\alpha}(t)]\right) \right\rbrace.   
}
\end{equation}

We want to find upper bounds to the modulus of all terms forming $D_t^{\alpha,I}$. The contribution due to $L$ can be bounded as 
$$
\left|\omega_t\left(L[m_\alpha^N-m_\alpha(t)]\right)\right|\le \frac{2C_L}{N}\, .
$$

The remaining ones are of the type (I) $\omega_t([m_{\beta}^{N} -m_{\beta}(t)]X[m_{\alpha}^{N} -m_{\alpha}(t)])$, and (II) $\omega_t( \left[ \Gamma^2_\ell(\Delta^{\ell}_N)-\Gamma^2_\ell(\Delta_{\ell}(t))\right] X[m_{\alpha}^{N} - m_{\alpha}(t)])$, with $X$ some operators. Terms such as (I) can be bounded as \cite{CarolloL:PRL:21}
\begin{equation}
|\omega_t([m_{\beta}^{N} -m_{\beta}(t)]X[m_{\alpha}^{N} -m_{\alpha}(t)])| \le \left\| X \right\| \E(t) \, ,
\end{equation}
and thus we have
\begin{eqnarray}
& |\omega_t([m_{\beta}^{N} -m_{\beta}(t)][m_{\alpha}^{N} -m_{\alpha}(t)])| \le  \E(t), \\
& |\omega_t([m_{\beta}^{N} -m_{\beta}(t)]m_{\gamma}^{N}[m_{\alpha}^{N} -m_{\alpha}(t)])| \le \E(t), \\
& |\omega_t(m_{\gamma}(t)[m_{\beta}^{N} -m_{\beta}(t)][m_{\alpha}^{N} -m_{\alpha}(t)])| \le \E(t),\\
& | \omega_t \left( \Gamma_{\ell}^{2}(\Delta_{\ell}(t))[m_{\beta}^{N} -m_{\beta}(t)][m_{\alpha}^{N} -m_{\alpha}(t)]  \right) | \le \gamma^2(\delta_\ell)\E(t)
\end{eqnarray}
with $X=\mathbb{I}$, $X=m_{\gamma}^{N}$, $X=m_{\gamma}(t)$, and $X = \Gamma_{\ell}^{2}(\Delta_{\ell}(t)) $, respectively. Note that we have further exploited that $\|m_{\alpha}^{N}\| \leq 1$, and, by Lemma \ref{lemma_bound_mf}, that $|m_{\alpha}(t)| \leq 1 $ and $|\Delta_\ell (t)| \leq \delta_{\ell}$.
For terms such as (II), we exploit Lemma \ref{Lemma-aux} and a Cauchy-Schwarz inequality to obtain 
\begin{equation}
\eqalign{
& |\omega_t(\left[ \Gamma^2_\ell(\Delta^{\ell}_N)-\Gamma^2_\ell(\Delta_{\ell}(t))\right]X[m_{\alpha}^{N} - m_{\alpha}(t)])| \le \\
&  2 \gamma(\delta_\ell)\gamma'(\delta_{\ell})\| X \|\sum_{\beta} |r_{\ell \beta}| \sqrt{\omega_t([m_{\beta}^N-m_{\beta}(t)]^2)}\sqrt{\omega_t([m_{\alpha}^N-m_{\alpha}(t)]^2)} \le \\
& 2 \delta_\ell \gamma(\delta_\ell)\gamma'(\delta_{\ell}) \| X \| \E(t) \, ,
}
\end{equation}
with $\| X \|=\| m_{\gamma}^{N}\| \le 1$. As a consequence, we can derive
\begin{equation}
|D_t^{\alpha,I} | \le \frac{C_0}{2}\E(t) +\frac{2C_L}{N}\, ,
\label{eq_def_c0}
\end{equation}
with 
$$
C_0=2\left(d^{2} A + d^42 B+ qd^2 M [\gamma^2(\delta)+2\delta \gamma(\delta)\gamma'(\delta)]\right)
$$
where $ A \equiv {\mathrm {max}}_{\alpha, \gamma}|A_{\alpha, \gamma}|$, $B \equiv {\mathrm {max}}_{\alpha, \beta, \gamma}|B_{\alpha \beta \gamma} |$, $M \equiv {\mathrm {max}}_{\ell, \alpha, \gamma} |M_{\ell \alpha}^{\gamma}|$, $\delta \equiv {\mathrm {max}}_{\ell}\delta_{\ell}$, and, recalling the definitions
\begin{equation*}
    \gamma(z):=\sum_{k=0}^\infty |c_\ell^k||z|^k <\infty\, ,\qquad \gamma'(z):=\sum_{k=1}^\infty k|c_\ell^k||z|^{k-1} <\infty\, ,
\end{equation*}
we have $\gamma(\delta)\ge\gamma(\delta_\ell)$, $\gamma'(\delta)\ge\gamma'(\delta_\ell)$ for all $\ell$.  Therefore, for the time derivative $D_t$, it is
\begin{equation}
|D_t |\le 2\sum_{\alpha=1}^{d^2}|D_t^{\alpha,I}|+|\omega_t\left(P\right)| \leq C_1 \, \E(t) + \frac{C_2}{N}\, ,
\end{equation}
where $C_1=d^2 C_0$ and $C_2=d^2(C_P+4 C_L)$. 

As a result, the time derivative of the cost function can be bounded by
\begin{equation}
\dot{\mathcal{E}}_{N}(t)  \le \left|\dot{\mathcal{E}}_{N}(t)   \right| \le C_1\E(t) + \frac{C_2}{N}\, .
\end{equation}
\end{proof}

\subsection{Open quantum dynamics of quasi-local operators}
\label{generator-quasi-local}
In this section, we show how the exactness of the mean-field equations can be exploited to derive the effective dynamical generator which implements the time evolution of any quasi-local operator in the thermodynamic limit.

Let us consider the following time-dependent Lindblad generator 
\begin{equation}
\tilde{\mathcal{L}}_t\left[\cdot \right]= i\left[\tilde{H},\cdot\right]+\sum_{\ell}\Gamma^2_\ell(\Delta_\ell(t))\mathcal{D}_\ell^{\rm Loc}\left[\cdot\right]\, ,
    \label{ql-gen}
\end{equation}
where $\mathcal{D}_\alpha^{\rm Loc}$ is the dissipator introduced in Eq.~\eref{D_loc}, $\Delta_\ell(t)$ is the linear combination of mean-field variables  
$$
\Delta_\ell(t)=\sum_\alpha r_{\ell\alpha} m_\alpha(t)\, , 
$$
and the Hamiltonian is given by
$$
\tilde{H}=\sum_{k=1}^N \sum_{\alpha=1}^{d^2} \epsilon_\alpha v_\alpha^{(k)}+\sum_{k=1}^N \sum_{\alpha,\beta=1}^{d^2} h_{\alpha \beta} \left(m_\alpha(t) v_\beta^{(k)}+m_\beta(t) v_\alpha^{(k)}\right)\, .
$$
Such a Lindblad operator is well-defined for any time $t$ since the variables $m_\alpha(t)$ are well-defined due to Lemma \ref{lemma_bound_mf}. 

Through this generator, we can define the (time-ordered) dynamical map $\Lambda_{t,s}^{N}\left[\cdot\right]$  such that
$$
\frac{d}{dt } \Lambda_{t,s}^N\left[\cdot\right]=\Lambda_{t,s}^N\circ \tilde{\mathcal{L}_t} \left[\cdot\right]\, , \qquad \mbox{ and } \qquad \frac{d}{ds } \Lambda_{t,s}^N\left[\cdot\right]=-\tilde{\mathcal{L}_s}\circ\Lambda_{t,s}^N  \left[\cdot\right]\, .
$$
Moreover, since this generator acts independently on the different particles, we also have that if an operator $O$ has support only on certain sites, then also $\Lambda_{t,s}^N\left[O\right]$ will have support on the same sites.  

As we shall prove in the following Theorem, the generator $\tilde{\mathcal{L}}_t$ is in fact the generator of the dynamics of any quasi-local observable, or operator, of the system, in the thermodynamic limit. This generator is thus valid for any operator in $\mathcal{A}$ and any sufficiently clustering state $\omega$ [see Definition \ref{clustering}]. Moreover, we note that the restriction of the state $\omega$ to a single particle or to a finite number of particles can be represented by a density matrix $\rho$. If the latter is in product form over the considered particles at the initial time, the structure of the generator $\tilde{\mathcal{L}}_t$, which acts separately on the different particles, guarantees that the time-evolved state $\rho_t$ for the finite set of particles considered will remain in product form for all times, in the thermodynamic limit. This is another way in which the validity of the mean-field approach can be understood, see e.g., the discussion in Ref.~\cite{Pickl11} where a closed many-body system is considered. Our observation here thus confirms that such an approach based on the study of the single-particle, or few-particle, reduced density matrix is in fact equivalent to the approach adopted in this work, which focuses on average operators as well as on quasi-local ones.

\begin{theorem}\label{theorem2}
The time evolution implemented by $\mathcal{L}_N$ converges, in the weak operator topology for the quasi-local algebra, to the time evolution implemented by the time-dependent generator $\tilde{\mathcal{L}}_t$ through the map $\Lambda_{t,0}^N$. That is, 
$$
\lim_{N\to\infty}\omega\left(A^\dagger e^{t\mathcal{L}_N}\left[O\right]B\right)=\lim_{N\to\infty}\omega\left(A^\dagger \Lambda_{t,0}^N\left[O\right]B\right)\, ,
$$
for all $A,B,O\in \mathcal{A}$ and any $t<\infty$. 
\end{theorem}

\begin{proof}
In order to keep a compact  notation, we define 
$$
\omega^{A^\dagger B}(O):=\omega(A^\dagger O B)\, ,
$$
and consider the difference 
$$
I_O=\omega^{A^\dagger B}\left(e^{t\mathcal{L}_N}\left[O\right]\right)-\omega^{A^\dagger B}\left(\Lambda_{t,0}^N\left[O\right]\right)\, .
$$
Let us first consider the case in which $O$ is strictly local and thus supported only on a finite number of sites. This means that there exists $k_{min}\le k_{max}$ such that $(k_{max}<\infty)$ 
$$
\left[v_\alpha^{(k)},O\right]=0\, , \forall v_\alpha\, , 
$$
whenever $k<k_{min}$ or $k>k_{max}$. The support of the operator then has length $N_O=k_{max}-k_{min}+1$. Since we focus here on the limit $N\to\infty$ we consider always $N>k_{max}$. 

We now start with the actual proof. The quantity $I_O$ can be written as 
$$
I_O=\omega^{A^\dagger B}\left(\int_0^t ds\, \frac{d}{ds} \left[e^{s\mathcal{L}_N}\circ \Lambda_{t,s}^N\left[O\right] \right]\right)\, .
$$
Calculating explicitly the derivative we find 
\begin{equation}\label{th2_I_O}
I_O=\int_0^t ds \, \omega^{A^\dagger B}\left(e^{s\mathcal{L}_N}\circ \left[\mathcal{L}_N-\tilde{\mathcal{L}}_s \right]\circ \Lambda_{t,s}^N[O]\right)\, .
\end{equation}
Now, let us define $O_{t,s}=\Lambda_{t,s}^N[O]$. This is an operator which has the same support of $O$, that is thus finite. 
Calculating the action of $\mathcal{L}_N$ on $O_{t,s}$ we find that 
\begin{equation}
\eqalign{
\mathcal{L}_N[O_{t,s}]= & \, i\sum_{\alpha}\epsilon_\alpha \sum_{k=k_{min}}^{k_{max}}\left[v_\alpha^{(k)},O_{t,s}\right] \\
& + i\sum_{\alpha \beta}h_{\alpha,\beta }\left(m_\alpha^N \sum_{k=k_{min}}^{k_{max}}\left[v_\beta^{(k)},O_{t,s}\right]+ \sum_{k=k_{min}}^{k_{max}}\left[v_\alpha^{(k)},O_{t,s}\right]m_\beta^N\right)\\ 
& + \sum_{\ell} \Gamma_\ell^2(\Delta_N^\ell) \mathcal{D}_\ell^{\rm Loc}\left[ O_{t,s}\right]+O(N)\, ,
}
\end{equation}
where the $O(N)$ takes into account the difference between the dissipative part of the original Lindblad and its local action modulated by the operator-valued rate which converges to zero in norm in the thermodynamic limit (cf.~Lemma \ref{Cor_diss_dyn}). With this, we can also calculate the difference between the action of the two generators appearing in Eq.~\eref{th2_I_O}. This is equivalent to 
\begin{equation}
\eqalign{
\mathcal{L}_N[O_{t,s}]-\tilde{\mathcal{L}}_s [O_{t,s}]=&i\sum_{\alpha \beta}h_{\alpha \beta }(m_\alpha^N-m_\alpha(s)) \sum_{k=k_{min}}^{k_{max}}\left[v_\beta^{(k)},O_{t,s}\right]+\\
&+i\sum_{\alpha \beta}h_{\alpha \beta }\sum_{k=k_{min}}^{k_{max}}\left[v_\alpha^{(k)},O_{t,s}\right](m_\beta^N-m_\beta(s))+\\ 
&+ \sum_{\ell} (\Gamma_\ell^2(\Delta_N^\ell)-\Gamma_\ell^2(\Delta_\ell(s))) \mathcal{D}_\ell^{\rm Loc}\left[ O_{t,s}\right]+O(N)\, .
}
\end{equation}
Now, we plug this back into the expression for $I_O$. We find
\begin{equation}
    \eqalign{
       I_O=&\int_0^t ds \, \omega^{A^\dagger B}\left(e^{s\mathcal{L}_N}\left[i\sum_{\alpha \beta}h_{\alpha \beta }(m_\alpha^N-m_\alpha(s)) \sum_{k=k_{min}}^{k_{max}}\left[v_\beta^{(k)},O_{t,s}\right]\right]\right) +\\
       &\int_0^t ds \, \omega^{A^\dagger B}\left(e^{s\mathcal{L}_N}\left[i\sum_{\alpha \beta}h_{\alpha \beta }\sum_{k=k_{min}}^{k_{max}}\left[v_\alpha^{(k)},O_{t,s}\right](m_\beta^N-m_\beta(s))\right]\right) +\\
       &\int_0^t ds \, \omega^{A^\dagger B}\left(e^{s\mathcal{L}_N}\left[\sum_{\ell} (\Gamma_\ell^2(\Delta_N^\ell)-\Gamma_\ell^2(\Delta_\ell(s))) \mathcal{D}_\ell^{\rm Loc}\left[ O_{t,s}\right]\right]\right) +\\
       &\int_0^t ds\, \omega^{A^\dagger B}\left(e^{s\mathcal{L}_N}\left[O(N)\right]\right) \, .
    }
    \label{terms-I}
\end{equation}
The last term of the above sum, which we call $I_O^4$ is easy to treat since $O(N)$ tends to zero in norm, and thus have $|I_O^4|\le \|A\|\|B\| \|O(N)\|$. The other terms need more care. Let us start with the first one, which we call $I_O^1$. Denoting 
$$
O_{t,s}^\alpha = \sum_{k=k_{min}}^{k_{max}}\left[v_\alpha^{(k)}, O_{t,s}\right]\, , 
$$
we can write
$$
I_O^1 = i\sum_{\alpha,\beta} h_{\alpha \beta}\int_0^t ds\, \omega^{A^\dagger B}\left(e^{s\mathcal{L}_N}\left[(m_\alpha^N-m_\alpha(s)) O_{t,s}^\beta\right]\right) \, .
$$
Using Lemma \ref{Lemma-dilation} proved in the Appendix (see also Ref. \cite{BenattiEtAl18}) we can write 
\begin{equation}
\eqalign{
|I_O^1|&\le d^4 h_{{\mathrm {max}}}\int_0^t ds\, \left|\omega^{A^\dagger B}\left(e^{s\mathcal{L}_N}\left[(m_\alpha^N-m_\alpha(s)) O_{t,s}^\beta\right]\right) \right|\le \\
& \le d^4 h_{{\mathrm {max}}} 2 N_O \|B\| \|v_\beta\|\|O\| \int_0^t ds\, \sqrt{\omega^{A^\dagger A}\left(e^{s\mathcal{L}_N}\left[(m_\alpha^N-m_\alpha(s))^2 \right]\right)}
}
\end{equation}
where we defined $h_{{\mathrm {max}}}={\mathrm {max}}_{\alpha, \beta} h_{\alpha \beta}$, and we further used that 
$$
\|O_{t,s}^\beta\|\le 2 N_O \|v_\beta\|\|O\|\, .
$$

Let us consider the term inside the square root. We can define a state obtained from $\omega$ through $A$ as
$$
\tilde{\omega}^{A^\dagger A}(X):=\frac{\omega\left(A^\dagger X A\right)}{\omega\left(A^\dagger A\right)}\,.
$$
The state $\tilde{\omega}^{A^\dagger A}$ is also clustering in the sense of Definition \ref{clustering}, as long as $\omega$ is clustering and $A$ is quasi-local. We thus have 
$$
\omega^{A^\dagger A}\left(e^{s\mathcal{L}_N}\left[(m_\alpha^N-m_\alpha(s))^2 \right]\right)=\omega(A^\dagger A) \tilde{\omega}^{A^\dagger A}\left(e^{s\mathcal{L}_N}\left[(m_\alpha^N-m_\alpha(s))^2 \right]\right)\, .
$$
As done in Theorem \ref{theorem}, we can define for the state $\tilde{\omega}^{A^\dagger A}$ the cost function 
$$
\mathcal{E}_N^A(t)=\sum_{\alpha}\tilde{\omega}^{A^\dagger A}\left(e^{t\mathcal{L}_N}\left[(m_\alpha^N-m_\alpha(t))^2 \right]\right)\, .
$$
Repeating all the steps in the proof of the theorem for this new cost function we find that 
$$
\mathcal{E}^A_N(t)\le e^{t C_1}\mathcal{E}^A_N(0)+\frac{1}{N}\frac{C_2}{ C_1}\left(e^{C_1 t}-1 \right) \, .
$$
The second term on the right-hand side goes to zero in the large $N$ limit. The first one also goes to zero since the state $\tilde{\omega}^A$ is clustering, in the sense that 
$$
\tilde{\omega}^{A^\dagger A}\left(\left[m_\alpha^N-m_\alpha(0) \right]^2\right)\to 0
$$
for $N\to\infty$. This can be seen by considering that average operators commute with quasi-local operators 
so that 
$$
\lim_{N\to\infty}\tilde{\omega}^{A^\dagger A}\left(\left[m_\alpha^N-m_\alpha(0) \right]^2\right)=\frac{1}{\omega(A^\dagger A)}\lim_{N\to\infty}{\omega}\left(A^\dagger A \left[m_\alpha^N-m_\alpha(0) \right]^2\right)\, .
$$
Using appropriately the Cauchy-Schwarz inequality, we then find 
$$
\lim_{N\to\infty}\tilde{\omega}^{A^\dagger A}\left(\left[m_\alpha^N-m_\alpha(0) \right]^2\right)\le 2 \|A\|^2 \lim_{N\to\infty}\sqrt{\omega\left(\left[m_\alpha^N-m_\alpha(0) \right]^2\right) }
$$
where we used $\|m_\alpha^N-m_\alpha(0)\|\le 2$. By assumption on the state $\omega$, the right-hand side of the above inequality goes to zero. 
All together this shows that 
$$
\lim_{N\to\infty}|I_O^1|=2d^4h_{{\mathrm {max}}} N_O \|B\|\|v_\beta\| \|O\|\omega(A^\dagger A) \lim_{N\to\infty}\int_0^t ds\, {\sqrt{\mathcal{E}_N^A(s)}}=0\, .
$$ 

The second term in Eq.~\eref{terms-I}, which we call $I_O^2$, can be treated exactly in the same way as above, so that we are left with the third term $I_O^3$. For the sake of clarity, we have
$$
I_O^3=\sum_{\ell}\int_0^t ds \, \omega^{A^\dagger B}\left(e^{s\mathcal{L}_N}\left[ (\Gamma_\ell^2(\Delta_N^\ell)-\Gamma_\ell^2(\Delta_\ell(s))) \mathcal{D}_\ell^{\rm Loc}\left[ O_{t,s}\right]\right]\right)\, , 
$$
and exploiting Lemma \ref{Lemma-aux} we immediately find
\begin{eqnarray*}
|I_O^3|\le & 2\sum_{\ell,\beta}|r_{\ell\beta}| \gamma_\ell(\delta_\ell)\gamma_\ell'(\delta_\ell)\sqrt{\omega(B^\dagger B)}\|\mathcal{D}_\ell^{\rm Loc}\left[ O_{t,s}\right]\| \\
& \times \int_0^t ds\, \sqrt{\omega^{A^\dagger A}\left(e^{s\mathcal{L}_N}\left[(m^N_\beta-m_\beta(s))^2\right]\right)}\, .
\end{eqnarray*}
Note that $\|\mathcal{D}_\ell^{\rm Loc}\left[ O_{t,s}\right]\|$ remains finite since the operator $O_{t,s}$ has local support. 
Now, looking at the square root inside the integral and recalling the discussion used to show that the term $I_O^1$ converges to zero, we can show that also $I_O^3$ is vanishing in the $N\to\infty$ limit. 

Collecting all these results together, we have shown that 
\begin{equation}
\lim_{N\to\infty }|I_O|=\lim_{N\to\infty}\left(|I_O^1|+|I_O^2|+|I_O^3|+|I_O^4|\right)=0\, , 
\label{theo-local}
\end{equation}
which concludes the proof of the theorem for any operator $O$ with strictly local support. 

To extend this result to the case of any operator $O$ of the quasi-local algebra $\mathcal{A}$, we proceed as follows.
First, we observe that any quasi-local operator $O$ can be approximated with arbitrary accuracy by an operator with local support. This means that for any $\varepsilon>0$, we can always find an operator $O_\varepsilon$, such that 
$$
\|O-O_\varepsilon \|\le \varepsilon\, .
$$
Then, we consider again the quantity $I_O$ and rewrite it as 
\begin{equation}
    \eqalign{
    I_O&=\omega^{A^\dagger B}\left(e^{t\mathcal{L}_N}[O]\right)-\omega^{A^\dagger B}\left(\Lambda_{t,0}^N[O]\right)=\\
    &=\omega^{A^\dagger B}\left(e^{t\mathcal{L}_N}[O-O_\varepsilon]\right)-\omega^{A^\dagger B}\left(\Lambda_{t,0}^N[O-O_\varepsilon]\right)+I_{O_\varepsilon}\, .
    }
\end{equation}
Since $e^{t\mathcal{L}_N}$ and $\Lambda_{t,0}^N$ are both contractions, the first two terms are bounded by $\|A\|\|B\|\|O-O_\varepsilon \|$. As such we have 
$$
|I_O|\le 2\|A\|\|B\|\varepsilon +|I_{O_\varepsilon}|\, .
$$
Now, we chose $\varepsilon$ such that $\varepsilon=\chi/(4\|A\|\|B\|)$. Moreover, 
due to the result for operators with strictly local support summarized by Eq.~\eref{theo-local}, given any $\chi>0$ there exists a $\tilde{N}$ such that for any $N>\tilde{N}$ we have 
$$
|I_{O_\varepsilon}|<\frac{\chi}{2}\, .
$$
With these considerations, we can thus say that $\forall \chi>0$ there exists a $\tilde{N}$ such that $\forall N>\tilde{N}$ we have 
$$
|I_O|< \chi\, ,
$$
which is nothing but the definition of the limit appearing in the theorem.
\end{proof}

\section{Application to quantum Hopfield-type neural networks}
\label{sec5}

\subsection{Quantum Hopfield-type neural networks}
In this section, we apply our results to open quantum generalizations of Hopfield-type models. These systems origin within the field of classical neural networks (NNs). They are fundamental models realizing associative memory behavior \cite{Amit_book}, i.e., they are capable of retrieving complete information from corrupted data, following a learning rule. 
A paradigmatic instance of associative memory is the so-called Hopfield neural network (HNN) \cite{Hopfield:1982}, that we will consider in the following. The HNN is a classical spin network featuring all-to-all interactions \cite{Amit_book,AmitGS:1985a}, described by the energy function $E=-\frac{1}{2}\sum_{i\neq j=1}^{N} w_{i j} \, \sigma_z^{(i)} \sigma_z^{(j)}$, where $N$ is the number of spins and $\sigma_z^{(i)}$ are classical Ising spins. The interaction couplings, $w_{ij}$, are chosen in such a way that a set of $p$ spin configurations, $\lbrace \xi_{i}^{\mu} \rbrace_{i = 1,...,N}$ for $\mu=1,2,\dots p$, can be stored and retrieved by the system which corresponds to the patterns being the minima of the energy function. The different spin configurations $\lbrace \xi^\mu_i\rbrace_i$ can represent patterns, such as images or letters of an alphabet. Among the different learning rules, widely known is the \emph{Hebb's prescription}, that sets $w_{i j} = \frac{1}{N}\sum_{\mu=1}^{p}\xi_{i}^{\mu} \xi_{j}^{\mu}$. 

In practice, for what concerns theoretical investigations on these models, the patterns can be chosen to be generated by independent identically distributed (i.i.d.) random variables that can assume the values $\xi_i^{\mu} = \pm 1$. For $p/N \ll 1$, the spin configurations which have minimal energy are those in which all spins are aligned with the patterns. 
The retrieval mechanism emerges when endowing the HNN with a Glauber thermal single spin-flip dynamics with inverse bath-temperature $\beta^{-1}$ \cite{glauber1963}.

Quantum generalizations of HNNs have been introduced in Refs.~\cite{Rotondo:JPA:2018,FiorelliLM22} to embed these systems into the more general framework of open quantum Markovian evolution defined by Eq.~\eref{e0_totalHamiltonian}-\eref{jumps} and to investigate the impact on quantum effects on their retrieval dynamics. Here, the system is described in terms of $N$ spin-$1/2$ particles, undergoing a Markovian evolution with jump operators
\begin{equation}\label{hopfield rates}
J_{\pm}^{(k)} = \sigma_{\pm}^{(k)}\Gamma_{\pm}^{\mathrm{HN}}(\Delta E_k), \qquad \Gamma_{\pm}^{\mathrm{HN}  } (\Delta E_k)= \frac{ e^{\pm \frac{\beta}{2}\Delta E_{k} }}{ \sqrt{2\cosh{(\beta \Delta E_k)}}} \, ,
\end{equation}
where $$ \Delta E_{k} = \frac{1}{N}\sum_{\mu=1}^{p}\xi^{\mu}_{k}\sum_{j \neq k} \xi_{j}^{\mu} \sigma^{(j)}_{z}$$ 
represents the energy difference associated with the configuration before and after the transition. We note that this operator quantifies the energy change associated with a spin-flip at site $j$. It is thus not a simple multiple of the identity but rather a many-body operator depending on the state of all spins. It is worth noticing that the operator-valued rates $\Gamma_{\ell}^{\mathrm{HN}}(\Delta E_k)$ for this model do not act on the $k$-th spin and thus commute with the operator $ \sigma_{\alpha}^{(k)}$, $\alpha= \pm$. The Hamiltonian term is chosen to be a homogeneous transverse field, $H=\Omega \sum_{i=1}^{N}\sigma_{x}^{(i)}$, and competes with the dissipative HNN dynamics. In the thermodynamic limit and for $p/N \ll 1$, the quantum model has been analyzed via the dynamical evolution of some macroscopic quantities. In fact, under a mean-field approximation, i.e., neglecting correlations among average operators, the retrieval properties of quantum HNNs in the parameter regimes $(\Omega, \beta)$ have been characterized \cite{RotondoEtal:2018}. It was shown that for large temperatures quantum HNNs display a so-called paramagnetic (disordered) phase, for which pattern retrieval is not possible. For small temperatures and for sufficiently small values of the transverse-field strength, the system shows instead a ferromagnetic phase and can operate as an associative memory. Interestingly, for large values of $\Omega$ a quantum retrieval phase can be observed, characterized by a limit-cycle regime in which the state of the system features a nonzero overlap with one of the patterns. However, while a mean-field theory holds true for the classical HNN, no exact proof is known for quantum generalizations of the model. We will now show that Theorem \ref{theorem} and Theorem \ref{theorem2} apply to these cases. 

Before proceeding, we introduce the quantity 
$$
\Delta E = \frac{1}{N}\sum_{\mu}\xi^{\mu}_{k}\sum_{j} \xi_{j}^{\mu} \sigma^{(j)}_{z}\, ,
$$
together with the operator-valued rates
\begin{equation}\label{hopfield_rates_global}
\Gamma_{\pm }^{\mathrm{HN}}(\Delta E) \equiv \frac{ e^{\pm \frac{\beta}{2}\Delta E }}{ \sqrt{2\cosh{(\beta \Delta E)}}}\, .
\end{equation}
The difference between these rates and the ones defined in Eq.~\eref{hopfield rates} is that the former also account for the self-energy contribution, i.e., the sum involves all sites. In the following, we will show that replacing the operator-valued rates \eref{hopfield rates} with the ones defined by \eref{hopfield_rates_global}, i.e., replacing $\Delta E_k$ with $\Delta E$, yields the same equation of motions in the thermodynamic limit. Nonetheless, the representation of the system dynamics via \eref{hopfield_rates_global} is more closely related to the results we presented above.

\subsection{Exactness of the mean-field approach}
It is worth noticing that the operator valued rates Eqs.~\eref{hopfield rates}, \eref{hopfield_rates_global} do not satisfy Assumption~\ref{Gamma}. Indeed, while given a real number $x$, $(\cosh{(\beta x)})^{-1/2}$ is a real analytic function, the function $(\cosh{(\beta z)})^{-1/2}$, with $z$ complex is not an entire function. This means that there is no power series for $(\cosh{(\beta x)})^{-1/2}$, centered in $x=0$ which has an infinite radius of convergence. Clearly, as we also mentioned just after Assumption~\ref{Gamma}, requiring that the operator-value rate functions admit a power series with infinite radius of convergence is a strong requirement which is not strictly necessary for our treatment. In fact, looking back at our proof of the theorems, we see that in order to exploit our theorems, the operator-valued functions $\Gamma_\ell$ must possess two fundamental requirements: i) they must obey an equivalent of Lemma \ref{lemma_commutators}, i.e., they must commute with local or average operators up to terms which must converge, in norm, to zero sufficiently fast with $N$ and ii)  they must obey an equivalent of Lemma \ref{Lemma-aux}, i.e., their convergence to their mean-field counterpart must be dominated by the (quadratic) convergence of average operators to their mean-field counterpart. As we show below, these requirements are satisfied by the rates in Eqs.~\eref{hopfield rates}- \eref{hopfield_rates_global}.

As we deal with a system of spin-$1/2$ particles, we consider as a basis of the single-particle algebra $M_2(\mathbb{C})$ the Pauli operators $\sigma_{x,y,z}^{(k)}$ and the identity. Since the rates are proportional to $\sigma_z$ at the different sites, we have $[\Gamma_{\ell}^{\mathrm{HN}}(\Delta E),\sigma_z^{(k)}] = 0$ $\forall i=1,...,N$ and $\ell =  \pm$. To show that these rates obey Lemma \ref{lemma_commutators}, we thus need to evaluate the norm of the commutator  $[\Gamma_{\ell}^{\mathrm{HN}}(\Delta E),\sigma_{+}^{(k)}]$. For doing this, we notice that $\sigma^{(k)}_{+} \Gamma_{\ell}^{\mathrm{HN}}(\Delta E) =  \Gamma_{\ell}^{\mathrm{HN}}(\Delta E-c\sigma^{(k)}_{z})\sigma^{(k)}_{+} $, where $|c|\le 2p/N$, since we have  $|w_{ij}| \leq p/N$. Therefore, we can write
\begin{equation}
\eqalign{
  &  \| [\Gamma_{\ell}^{\mathrm{HN}}(\Delta E),\sigma^{(k)}_{+}] \| \leq \|   \Gamma_{\ell}^{\mathrm{HN}}(\Delta E) -  \Gamma_{\ell}^{\mathrm{HN}}(\Delta E-c\sigma^{(k)}_{z})\| \, .
}
\label{bound-HNN}
\end{equation}
It is straightforward to check that given two real values $x,y$, one has 
\begin{equation}
\eqalign{
\Gamma_{\pm}^{\mathrm{HN}}(x) -  \Gamma_{\pm}^{\mathrm{HN}}(y)&=\frac{e^{\pm \frac{\beta x}{2}}\left[\cosh(y)-\cosh(x)\right]}{\sqrt{\cosh(x)\cosh(y)}\left[\sqrt{2\cosh(y)}+\sqrt{2\cosh(x) }\right]} +\\
&+\frac{e^{\pm \frac{\beta x}{2}}-e^{\pm \frac{\beta y}{2}}}{\sqrt{2\cosh(y)}}\, ,
}
\end{equation}
which shows that the difference between the functions can be bounded by the sum of the difference between two entire functions multiplied by bounded terms (since $\Delta E$ only has finite eigenvalues). As such, in the right hand side of Eq.~\eref{bound-HNN} we can expand the entire functions in their Taylor series (with infinite radius of convergence), and noticing that the difference between the arguments of the two functions is of order $1/N$, one can find a suitable constant $C_\sigma$ such that 
$$
 \| [\Gamma_{\ell}^{\mathrm{HN}}(\Delta E),\sigma^{(k)}_{+}] \|\le \frac{C_\sigma}{N}\, .
$$
The same result holds for  $[\Gamma_{\ell}^{\mathrm{HN}}(\Delta E),\sigma^{(k)}_{-}]$. Therefore, for any single-site operator we have that $\|[\Gamma_{\ell}^{\mathrm{HN}}(\Delta E),x^{(k)}]\|\sim 1/N$ and thus also for average operators for which we have 
\begin{equation}\label{bound1:Hopfield}
     \| [\Gamma_{\ell}^{\mathrm{HN}}(\Delta E),X_N] \| \leq 2 x_{{\mathrm {max}}}  \frac{C_{\sigma}}{N}\, , 
\end{equation}
where $x_{{\mathrm {max}}}$ is the modulus of the matrix element of $x$ with largest absolute value. 
We further note that, since the Pauli matrices $\sigma_\pm,\sigma_z$ form a basis for the single-site algebra, any local operator $O$ can be written as a finite linear combination of products of Pauli matrices solely acting non-trivially on the support of $O$. Using our results on single-site operators and the linearity of the commutator we can thus find a bound for any strictly local operator $O\in\mathcal{A}$.

We now consider the double commutator $[\Gamma_{\ell}^{\mathrm{HN}}(\Delta E),[\Gamma_{\ell}^{\mathrm{HN}}(\Delta E), \sigma_{\alpha}^{(k)}]]$. Clearly this vanishes when considering $\sigma_\alpha = \sigma_z$ or equal to the identity. When focusing on $\sigma^{(k)}_{+}$, we have by direct computation,
\begin{equation}
\eqalign{
  &  \| [\Gamma_{\ell}^{\mathrm{HN}}(\Delta E),[\Gamma_{\ell}^{\mathrm{HN}}(\Delta E),\sigma_+^{(k)}]] \| \leq \|   \Gamma_{\ell}^{\mathrm{HN}}(\Delta E) -  \Gamma_{\ell}^{\mathrm{HN}}(\Delta E-c\sigma^{(k)}_{z})]\|^2  \, ,
}
\end{equation}
which, because of the argument above, is of order $1/N^2$. We thus find
\begin{equation}\label{bound2:Hopfield}
\eqalign{
  &  \| [\Gamma_{\ell}^{\mathrm{HN}}(\Delta E),[\Gamma_{\ell}^{\mathrm{HN}}(\Delta E),X_N]] \| \leq  2 x_{{\mathrm {max}}}  \frac{C_{\sigma}^2}{N^2} \, .
}
\end{equation}
An analogous result can be obtained again for strictly local operators. The results in Eq.~\eref{bound1:Hopfield} and in Eq.~\eref{bound2:Hopfield} can then be directly employed to show the validity of an equivalent of Lemma \ref{Cor_diss_dyn} for the operator-valued rates describing the quantum Hopfield NN. Indeed, by proceeding according to the proof of the above mentioned Lemma, and exploiting the specific form of the jump operators as defined by Eq.~\eref{hopfield_rates_global}, we find 
\begin{equation}\label{dissipative_effective_1}
\| \mathcal{D}_{\ell}[x_k] -  \Gamma_{\ell}^{ \mathrm{HN}\, 2}(\Delta E) \mathcal{D}_{\ell}^{\mathrm{Loc}}[x_k] \| \leq  \frac{x_{{\mathrm {max}}} }{N} \Gamma ( 3 C_{\sigma} +C_{\sigma} ^{2} ), 
\end{equation}
where $\Gamma $ is the norm of the operator-valued rate $\Gamma_{\ell}^{\mathrm{HN}}(\Delta E)$. 

One can check that the operator-valued rate defined by Eq.~\eref{hopfield rates} obeys the same bounds shown above. Hence, when considering $\Gamma_{\ell}^{\mathrm{HN}}(\Delta E_i)$, and identifying with $\tilde{\mathcal{D}}_{\ell}[\cdot]$ the corresponding dissipator, we get 
\begin{equation}\label{dissipative_effective_2}
\| \tilde{\mathcal{D}}_{\ell}[x_k] -  \Gamma_{\ell}^{ \mathrm{HN}\, 2}(\Delta E_i) \mathcal{D}_{\ell}^{\mathrm{Loc}}[x_k] \| \leq  \frac{x_{{\mathrm {max}}} }{N}  \Gamma  ( 3 C_{\sigma} +C_{\sigma}^{2} )\, . 
\end{equation}
At this point, we can show that in the thermodynamic limit the operator-valued rates \eref{hopfield rates} and \eref{hopfield_rates_global} give rise to the same equations of motions. To this end, let us consider the difference
\begin{equation}
\eqalign{
D= &  \| \mathcal{D}_{\ell}[x_h]  -  \tilde{\mathcal{D}}_{\ell}[x_h] \|  \\ 
= & \| \left( (\Gamma_{\ell}^{ \mathrm{HN}}(\Delta E) )^{2} - (\Gamma_{\ell}^{ \mathrm{HN}}(\Delta E_k))^{2} \right) \mathcal{D}_{\ell}^{\mathrm{Loc}}[x_h] + L_1 + L_2 \| \, ,
 }
\end{equation}
where $L_1 $  and $L_2$ follow the bounds as given by Eqs. \eref{dissipative_effective_1}, \eref{dissipative_effective_2}, respectively. Let us then focus on the norm of the remaining term. Noticing that $\Delta E_k=\Delta E - \frac{c}{2}\sigma_z^{(k)}$, we consider the difference
\begin{equation}
\eqalign{
    &  \left( (\Gamma_{\ell}^{ \mathrm{HN}}(\Delta E) )^{2} - (\Gamma_{\ell}^{ \mathrm{HN}}(\Delta E - \frac{c}{2}\sigma_z^{(k)}))^{2} \right) \mathcal{D}_{\ell}^{\mathrm{Loc}}[x_h]    \\
    & = \left( e^{\ell \beta \Delta E} \cosh{(\ell \beta \Delta E - \frac{c}{2}\sigma_z^{(k)})} -e^{\ell \beta (\Delta E - \frac{c}{2}\sigma^{(k)}_{z})} \cosh{(\ell \beta \Delta E )}  \right) W \, ,
}
\end{equation}
where $W = \mathcal{D}_{\ell}^{\mathrm{Loc}}[x_k]/(2 \cosh{(\beta \Delta E)} \cosh{(\beta \Delta E_k)})$ is a norm-bounded operator. We can thus focus on the norm
\begin{equation}
\eqalign{
    & \| e^{\ell \beta \Delta E} \cosh{(\ell \beta \Delta E - \frac{c}{2}\sigma_z^{(k)})} -e^{ \ell \beta (\Delta E - \frac{c}{2}\sigma_{z}^{(k)})} \cosh{(\ell \beta \Delta E )}  \| \\
    & \leq \| \sinh{(\beta c \sigma_z^{(k)} /2 )}\| \leq \sum_{n=0}^{\infty} \frac{1}{(2n+1)!} (\beta p)^{(2n+1)} \frac{1}{N^{(2n+1)}} \\
    & \leq \frac{1}{N} \sum_{n=0}^{\infty} \frac{1}{(2n+1)!} (\beta p)^{(2n+1)} \leq \frac{1}{N} \sinh{(\beta p)} \, ,}
\end{equation}
having exploited the convergence of the series $\sum_n \frac{1}{(2n+1)!} |b|^n < \infty$, $\forall b \in \mathbb{R}$.
As a result, we derive 
\begin{equation}
     \| \mathcal{D}_{\ell}[x_h]  -  \tilde{\mathcal{D}}_{\ell}[x_h] \| \leq \frac{1}{N}[\sinh(\beta p) \| W \| + 2x_{{\mathrm {max}}} \Gamma  ( 3 C_{\sigma} +C_{\sigma} ^{2} )]
\end{equation}

We now explicitly show that the operator-valued rates in Eq.~\eref{hopfield_rates_global} can be written as a linear combination of suitable average operators. To this end, we perform a mapping \cite{CarolloL:PRL:21} [see Fig.~\ref{Fig2}] on the all-to-all classical energy function
\begin{equation}
E=-\frac{1}{2}\sum_{i,j} w_{ij} \sigma_z^{(i)} \sigma_z^{(j)} = -\frac{1}{2 N}\sum_{\mu=1}^{p}  \left( \sum_{i=1}^{N}\xi_i^{\mu} \sigma_z^{(i)} \right)^2,
\end{equation}
where the expression of $w_{ij}$ in terms of the patterns $\xi_i^{\mu}$ has been written. We will now reorder the $p$ rows of the patterns $(\xi_1^{\mu}, ...,\xi_{N}^{\mu}) $, each one corresponding to a $ \lbrace \sigma_z^{(i)} \rbrace_{i=1,..., N}$ spin configuration. The first pattern, $\xi_i^{1}$ takes the values $\pm 1$ at random positions. We relabel the spins as follows: the ones for which $\xi_{j}^1 = +1$ are taken to the left, and the remaining ones, for which $\xi_{j}^1 = -1$, to the right, as shown in Fig.~\ref{Fig2}. Thus, there exists $\tilde{h}$ such that $\xi_h^1=1$ for $h\leq \tilde{h}$, and $\xi_h^1=-1$ otherwise. Next, we consider the second pattern, $\xi_{i}^{2}$. In the subset corresponding to $\xi_h^1=1$ we relabel the spins such that $\xi_i^{2}=1$ are moved to the left, and $\xi_i^{2}=-1$ are moved to the right. The same is done for the subset corresponding to $\xi_h^1=-1$. This procedure can be repeated up to the last pattern. For large $N$, such a mapping yields $2^p$ subset of spins, pictorially illustrated in Fig.~\ref{Fig2}, each one described by macroscopic spin operators, that interact among each other. In the following, we will denote these subsets as $\Lambda_k$, $k=1,...,2^p$. Furthermore, being $\xi_i^{\mu}$ i.i.d.~random variables, and so long as $N \gg 1 $, each pattern $(\xi_1^{\mu},...,\xi_{N}^{\mu})$  contains, at leading order, an equal number of $+1$ and $-1$. Thus each one of the $2^p$ subsets has at leading order the same number of spins, $N_{\mathrm{s}}=N / 2^p$ (assuming that $N/2^p$ is an integer number). Under this mapping, the energy function reads
\begin{equation}
E= -\frac{1}{2}\sum_{h,k=1}^{2^p N_{\mathrm{s}}}\tilde{w}_{hk} S^{(h)}_{z}S^{(k)}_z = -\frac{1}{ 2^{p+1}N_{\mathrm{s}}}\sum_{\mu=1}^{p}\left( \sum_{h=1}^{2^{p}}f_{h}^{\mu} S^{(h)}_z \right)^2,
\end{equation}
which describes the interaction between large-spin operators $S^{(h)}_z$, where $S^{(h)}_{\alpha}= \sum_{i \in \Lambda_{h}} \sigma^{(i)}_{\alpha}$ is defined by the sum of spin-$1/2$ operators belonging to the $h$-th subset $\Lambda_h$. The coefficients $f_{h}^{\mu}$ [cf.~Fig.~\ref{Fig2}], which can assume the values $\pm 1$, represent the pattern values for spins in the subset $\Lambda_h$.  Most notably, they enter the definition of $\tilde{w}_{hk}=\frac{1}{ 2^{p}} \sum_{\mu=1}^{p} f_{h}^{\mu}f_{k}^{\mu}$, which specifies the interaction coupling between the $k$-th and $h$-th large spins. Furthermore, when considering the set of spin $1/2$ belonging to the set $\Lambda_k$, the operator $\Delta E$  becomes
\begin{equation}
\Delta E_{ \Lambda_k}^{N_{s}} = \frac{1}{ N_{\mathrm{s}}} \sum_{h=1}^{2^p}\tilde{w}_{hk} S^{(h)}_{z} =  \sum_{h=1}^{2^p} \tilde{w}_{hk} \, m_{z,h}^{N_{\mathrm{s}}}
\end{equation}
where we introduced the average magnetization operator
\begin{equation}\label{macro_spin_average}
m_{\alpha,k}^{N_{\mathrm{s}}} \equiv \frac{S^{(k)}_{\alpha}}{N_{\mathrm{s}}} \, ,
\end{equation}
for $\alpha=x,y,z,$ and $k=1,...,2^p$. The mapping clarifies how to derive the average operator description that we have employed for deriving a mean field description.

\begin{figure}[t]
\centering
\includegraphics[width=0.8\textwidth]{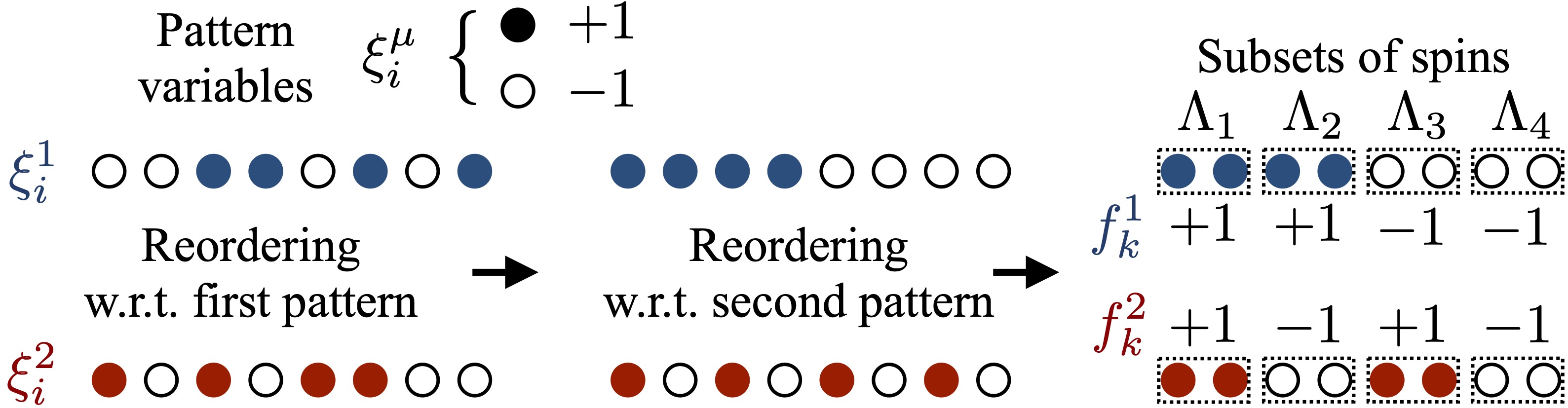}
\caption{{\bf Sketch of the mapping to large spins.} Example of the mapping discussed in the main text for $N=8$ spins and $p=2$ patterns. Each of the variables $\xi_i^\mu$ composing the patterns can assume either the value $+1$ or the value $-1$. The first step of the mapping consists in permuting the spins in a way that, after the transformation, the first pattern has all $\xi_i^1=+1$ appearing before the $\xi_i^1=-1$. This reshuffles also the structure of the second pattern. In the second step, we permute the spins inside the two sub-blocks identified by the transformed first pattern $\xi_i^1$. Spins are reordered in such a way that the second pattern $\xi_i^2$ has the values $+1$ appearing before the values $-1$ in each of the sub-block identified by the first pattern. This procedure generates $2^p$ subsets of spins $\Lambda_k$ (which for large $N$ form large-spin subsystems), such that if $m,n\in\Lambda_k$, then $\xi_m^\mu=\xi_n^\mu=f_k^\mu$, $\forall \mu$.  Here, the function $f_k^\mu$ is a representation of the pattern $\xi^\mu_i$ in terms of the subsets $\Lambda_k$. }
\label{Fig2}
\end{figure} 

We will now conclude and show that Theorem \ref{theorem} can be applied in this case. To this end, we recall that we have recovered the result of Lemma \ref{lemma_commutators} and \ref{Cor_diss_dyn} (with Lemma \ref{lemma_gen_action} and \ref{lemma_bound_mf} being actually independent of Assumption \ref{Gamma}). Hence, we need only to verify that an equivalent of Lemma \ref{Lemma-aux} holds true. It is sufficient to write
\begin{equation}
\eqalign{
 &  \Gamma_{\ell}^{\mathrm{HN}\, 2}(\Delta E_{\Lambda_k}^{N_{\mathrm{s}}})-\Gamma_{\ell}^{ \mathrm{HN}\, 2 }(\Delta E_{\Lambda_k}(t)) \\ 
 & =  \frac{e^{\ell \beta \Delta E_{\Lambda_k}^{N_{\mathrm{s}}}}\cosh(\beta \Delta E_{\Lambda_k}(t)) - e^{\ell \beta \Delta E_{\Lambda_k}(t)} \cosh(\beta \Delta E_{\Lambda_k}^{N_{\mathrm{s}}})}{2\cosh(\beta \Delta E_{\Lambda_k}^{N_{\mathrm{s}}})\cosh(\beta \Delta E_{\Lambda_k}(t))}  \\
 & = \sinh(\ell \beta (\Delta E_{\Lambda_k}^{N_{\mathrm{s}}} - \Delta E_{\Lambda_k}(t)))\frac{1}{2 \cosh(\beta \Delta E_{\Lambda_k}^{N_{\mathrm{s}}})\cosh(\beta \Delta E_{\Lambda_k}(t))}  \\
& = \sum_{h}\tilde{w}_{h k }(m_{z, h}^{{N_{\mathrm{s}}}} - m_{z, h}(t))Q_{k} \, ,
}
\end{equation}
i.e.~the difference of the operator-valued rates evaluated on averaged operator $\Delta E^{N_{\mathrm{s}}}$ and on the linear combination of mean-field variables $\Delta E(t)$, is dominated by an entire function, $\sinh(\cdot)$ times a norm-bounded one.  The power series expansion of the former is employed, and all the norm-bounded, remaining part is kept in $Q_{k}$, that reads
\begin{equation}
    Q_k = R \sum_{n=0}^{\infty} \frac{1}{(2n+1)!} (\ell \beta) ^{2n+1} \left( \sum_{h} \tilde{w}_{hk} (m_{z, h}^{{N_{\mathrm{s}}}}-m_{z, h}(t)) \right)^{2n}  ,
\end{equation}
where $R = 1/({2 \cosh{(\beta \Delta E_{\Lambda_k}^{N_{\mathrm{s}}})} \cosh{(\beta \Delta E_{\Lambda_k}(t))}})$. Thus the norm of $Q_k$ can be bounded as
\begin{equation}
\eqalign{
\|Q_k \| & = \|R \|\left\| \sum_{n=0}^{\infty} \frac{1}{(2n+1)!} (\ell \beta) ^{2n+1} \left( \sum_{h} \tilde{w}_{hk} (m_{z, h}^{{N_{\mathrm{s}}}}-m_{z, h}(t)) \right)^{2n}   \right\| \\
& \leq \| R \| \sum_{n=0}^{\infty} \frac{1}{(2n+1)!} \beta ^{(2n+1)} \left( \sum_{h} 2|\tilde{w}_{hk}| \right)^{2n} \\
& \leq\frac{ \| R \| }{ 2 \delta_{ E_{\Lambda_k}}}\sum_{n=0}^{\infty} \frac{1}{(2n+1)!} (2\beta\delta_{E_{\Lambda_k}})^{2n+1} = \frac{ \| R \| }{ 2 \delta_{E_{\Lambda_k}}} \sinh(2\beta   \delta_{E_{\Lambda_k}} ) \, ,}
\end{equation}
%. 
having exploited $\| \Delta E_{\Lambda_k} \| \leq \sum_{h} |\tilde{w}_{hk}| \equiv \delta_{E_{\Lambda_k}}$.
In this way, the proof of the lemma can be retraced, yielding 
\begin{equation}
\eqalign{
    &| \omega(A^{\dagger} e^{t \linn} [(\Gamma^{2}_{\ell}(\Delta E_{\Lambda_{k}^{{N_{\mathrm{s}}}}}) - \Gamma^{2}_{\ell}(\Delta E_{\Lambda_{k}}(t)))X] B) | \\
    & \leq C \| X \| \sum_{h} |\tilde{w}_{hk} | \sqrt{\omega(A^{\dagger } e^{t \linn}[(m_{z, h}^{{N_{\mathrm{s}}}}-m_{z, h}(t))^2] A)} \sqrt{\omega(B^{\dagger }B)} \, ,
    }
\end{equation}
with $C = \| R \| \sinh(2\beta \delta_{E_{\Lambda_{k}}})/(2 \delta_{E_{\Lambda_{k}}})$.
As such, Theorem \ref{theorem} can be applied. Using also the results on the commutator of the rates with local operators, Theorem \ref{theorem2} can be proved as well for the quantum generalization of the HNN dynamics. 

\section{Conclusions}

In this manuscript we considered many-body open quantum systems that evolve under a dynamical generator written in Lindblad form. We introduced the dissipative part of the latter as a generalization of classical stochastic dynamical generators where single-site transitions occur at a rate that depends on collective properties of the system itself. In the quantum setting, these are represented by operator-valued functions, $\Gamma$, assumed to be (real) analytic functions of average operators. We then added the coherent part of the dynamics by means of a single-particle Hamiltonian and an all-to-all two-body interacting Hamiltonian. Firstly we showed that, for large system size, the dissipative map on strictly local and average operators acts as a local dissipative map, weighted by the square of the operator-valued rates (Lemma \ref{Cor_diss_dyn}). We then moved forward to analyze the dynamics of average operators in terms of their Heisenberg equations. In fact, Theorem \ref{theorem} shows that the latter are exactly given by the mean-field equations of motions (given by factorizing expectation values of operators) in the thermodynamic limit. This is our second main result. Thirdly, we focused on the dynamics of quasi-local operators. Here, starting from the exactness of the mean-field equations, we derived the effective dynamical generator which provides their dynamics, in the thermodynamic limit (Theorem \ref{theorem2}). Finally, we showed the relevance of this results for the class of open quantum-Hopfield models in the limit of large system size and vanishing storage-capacity. 

It would be interesting to modify our approach in order to investigate a more general form of Lindblad operators, e.g., going beyond a collective all-to-all coupling Hamiltonian or permutation-invariant models. For instance, one could consider a translation-invariant Hamiltonian with two-body interactions and a generic translation-invariant dissipator, such as those emerging in the presence of light-mediated interactions (see, e.g., Ref.~\cite{williamson2020}), and develop an approach to these models by analyzing the dynamical behaviour of a suitable (possibly infinite) set of average operators defined in Fourier space. In contrast to all-to-all coupling models, where only the zero Fourier modes are relevant, translation-invariant systems require the consideration of all modes. 

\section{Acknowledgments}

EF and MM acknowledge support by the ERC Starting Grant QNets through Grant Number 804247. EF, MM and IL are grateful for funding from the Deutsche Forschungsgemeinschaft (DFG) through Grant No. 449905436. We also acknowledge funding by the DFG through the Research Unit FOR 5413/1, Grant No. 465199066 and under Germany’s Excellence Strategy – EXC-Number 2064/1 – Project number 390727645. FC~is indebted to the Baden-W\"urttemberg Stiftung for the financial support of this research project by the Eliteprogramme for Postdocs.

\appendix
\setcounter{lemma}{0}

\section{Lemmata}
\subsection{Proof of Lemma 1} \label{app_lemma_comm}

\begin{lemma}
If the function $\Gamma_\ell(\Delta_N^\ell)$ obeys Assumption \ref{Gamma}, then 
\begin{eqnarray*}
	&i) \left\| \left[\Gamma_\ell(\Delta_N^\ell), O\right] \right\| \le \frac{2N_O}{N}\|O\|\delta_\ell \gamma'(\delta_\ell) \, ,\\
	&ii) \left\| \left[\Gamma_\ell(\Delta_N^\ell), X_{N}\right] \right\| \le \frac{2}{N}\|x\| \delta_\ell  \gamma'(\delta_\ell)\, ,\\
	&iii) \left\| \left[\Gamma_\ell(\Delta_N^\ell),\left[\Gamma_\ell(\Delta_N^\ell), O\right]\right] \right\| \le \frac{4N_O^2}{N^2}\|O\|\delta_\ell^2 [\gamma'(\delta_\ell)]^2  \, , \\
	& iv)  \left\| \left[\Gamma_\ell(\Delta_N^\ell),\left[\Gamma_\ell(\Delta_N^\ell), X_N\right]\right] \right\| \le \frac{4}{N^2}\|x\|\delta_\ell^2 [\gamma'(\delta_\ell)]^2    \, ,
\end{eqnarray*}
with $O$ being any operator with strictly local support, $N_O$ the length of such support, and $X_N$ any average operator as defined in Eq.~\eref{eq:average-operators}. 
\end{lemma}

\begin{proof}
Given an operator $O$ which is supported only on a finite number of sites, we can always find two integer numbers $k_{min}\le k_{max}<\infty$ defining its support. In particular, $k_{min}$ is the largest number for which 
$$
\left[v_\alpha^{(k)},O\right]=0\, , \qquad \forall v_\alpha ,
$$
whenever $k<k_{min}$. The integer number $k_{max}$ is instead the smallest one for which 
$$
\left[v_\alpha^{(k)},O\right]=0\, , \qquad \forall v_\alpha ,
$$
for all  $k>k_{max}$. We then say that the operator $O$ has support which extends from site $k_{min}$ to site $k_{max}$ and that $N_O=k_{max}-k_{min}+1$ is the length, or extension, of its support. 

With this observation we can proceed with the proof of $i)$. This is done by directly evaluating the commutator. We have that 
$$
\left[\Gamma_\ell\left(\Delta_N^\ell\right), O\right]=\sum_{n=0}^\infty c_\ell^n \left[\left(\Delta_N^\ell\right)^n,O\right]=\sum_{n=0}^\infty c_\ell^n \sum_{j=0}^{n-1}\left(\Delta_N^\ell\right)^j\left[\Delta_N^\ell, O\right]\left(\Delta_N^\ell\right)^{n-1-j}\, .
$$
Next we evaluate the commutator of $\Delta_N^\ell$ and the local operator $O$. This gives 
$$
\left[\Delta_N^\ell, O\right]=\frac{1}{N}\sum_{k=1}^N\left[\sum_{\alpha=1}^{d^2}r_{\ell\alpha}v_\alpha^{(k)},O\right]\, .
$$
Because of the locality of the operator $O$, we further find 
$$
\left[\Delta_N^\ell, O\right]=\frac{1}{N}\sum_{k=k_{min}}^{k_{max}}\left[\sum_{\alpha=1}^{d^2}r_{\ell\alpha}v_\alpha^{(k)},O\right]\, .
$$
We then define the operator $O^\ell:=\left[\Delta_N^\ell, O\right]$ which, because of the above observation, is a local operator supported on the same sites of $O$ and with norm 
$$
\|O^\ell\|\le \frac{2N_O}{N}\|O\|\delta_\ell\, ,
$$
where $\delta_\ell$ is defined by Eq.~\eref{delta_alpha}. Plugging back this information in the commutator $i)$ and taking appropriate norm bounds (recall that $\|\Delta_N^\ell\|\le \delta_\ell$) we have
\begin{equation}
\|\left[\Gamma_\ell\left(\Delta_N^\ell\right), O\right]\|\le \frac{2N_O}{N}\|O\|\delta_\ell \gamma'(\delta_\ell)
\label{check_i)}
\end{equation}
where the quantity $\gamma'(\delta_\ell)$ is defined as the series 
$$
\gamma'(\delta_\ell)=\sum_{n=0}^\infty |c_\ell^n|n\, \delta_\ell^{n-1} <\infty\, .
$$
The relation in Eq.~\eref{check_i)} is exactly relation $i)$ reported in the Lemma. 

We now prove $iii)$ using some of the previous results. Considering the commutator in $i)$ we have already shown that 
$$
\left[\Gamma_\ell(\Delta_N^\ell), O\right]=\sum_{n=0}^\infty c_\ell^n \sum_{j=0}^{n-1}(\Delta_N^\ell)^j O^\ell (\Delta_N^\ell)^{n-1-j}\, .
$$
Now, to prove $iii)$ we need to consider a further commutator with $\Gamma_\ell(\Delta_N^\ell)$. Using the power series definition of $\Gamma_\ell(\Delta_N^\ell)$, the double commutator can be written as 
\begin{eqnarray*}
& \left[\Gamma_\ell(\Delta_N^\ell),\left[\Gamma_\ell(\Delta_N^\ell), O\right]\right] \\
& =\sum_{n,n'=1}^\infty c_\ell^n c_\ell^{n'} \sum_{j=0}^{n-1}(\Delta_N^\ell)^j \left[\sum_{i=0}^{n'-1}(\Delta_N^\ell)^i \left[\Delta_N^\ell,O^\ell\right](\Delta_N^\ell)^{n'-1-i} \right](\Delta_N^\ell)^{n-1-j}\, .
\end{eqnarray*}
Now we focus on the operator $O^{\ell \ell}:=[\Delta_N^\ell,O^\ell]$. Expanding $\Delta_N^\ell$, we can write
$$
O^{\ell\ell}=\frac{1}{N}\sum_{k=1}^N \left[\sum_{\alpha=1}^{d^2}r_{\ell\alpha} v_\alpha^{(k)},O^\ell\right]=\frac{1}{N}\sum_{k=k_{min}}^{k_{max}} \left[\sum_{\alpha=1}^{d^2}r_{\ell\alpha} v_\alpha^{(k)},O^\ell\right]\, ,
$$
where in the second equality we used the fact that $O^\ell$ is a strictly local operator with same support as $O$. This shows that 
$$
\|O^{\ell\ell}\|\le \frac{2N_O}{N}\|O^\ell\|\delta_\ell \le \frac{4N_O^2}{N^2}\|O\|\delta_\ell^2 \, ,
$$
which we can use to find the bound in $iii)$ 
$$
\left\|\left[\Gamma_\ell(\Delta_N^\ell),\left[\Gamma_\ell(\Delta_N^\ell), O\right]\right]\right\|\le \frac{4N_O^2}{N^2}\|O\|\delta_\ell^2 \left[\gamma'(\delta_\ell)\right]^2\, . 
$$

Now, we can straightforwardly prove relation $ii)$ and $iv)$ using $i)$ and $iii)$. For $ii)$ we consider that 
$$
\left\|\left[\Gamma_\ell(\Delta_N^\ell),X_N\right]\right\|
\le \frac{1}{N}\sum_{k=0}^N\left\|\left[\Gamma_\ell(\Delta_N^\ell),x^{(k)}\right]\right\|\, .
$$
Now, the norm of the commutator on the right-hand side does not really depend on $k$ due to the permutation invariance of the operator $\Gamma_\ell(\Delta_N^\ell)$ so that we have
$$
\left\|\left[\Gamma_\ell(\Delta_N^\ell),X_N\right]\right\|\le \left\|\left[\Gamma_\ell(\Delta_N^\ell),x^{(k)}\right]\right\|\, .
$$
We can exploit the result of $i)$, noticing that $x^{(k)}$ is a local operator with support equal to $N_{x^{(k)}}=1$, to find 
$$
\left\|\left[\Gamma_\ell(\Delta_N^\ell),X_N\right]\right\|\le \frac{2}{N}\|x\|\delta_\ell\gamma'(\delta_\ell)\, .
$$
We can proceed in a similar way for $iv)$. Indeed, we have 
\begin{eqnarray*}
\left\|\left[\Gamma_\ell(\Delta_N^\ell),\left[\Gamma_\ell(\Delta_N^\ell), X_N\right]\right]\right\|&\le \frac{1}{N} \sum_{k=1}^N \left\|\left[\Gamma_\ell(\Delta_N^\ell),\left[\Gamma_\ell(\Delta_N^\ell), x^{(k)}\right]\right]\right\|\\
&\le \left\|\left[\Gamma_\ell(\Delta_N^\ell),\left[\Gamma_\ell(\Delta_N^\ell), x^{(k)}\right]\right]\right\|\, ,
\end{eqnarray*}
and since $x^{(k)}$ is local, exploiting $iii)$ we can conclude that 
$$
\left\|\left[\Gamma_\ell(\Delta_N^\ell),\left[\Gamma_\ell(\Delta_N^\ell), X_N\right]\right]\right\|\le \frac{4}{N^2}\|x\|\delta_\ell^2 [\gamma'(\delta_\ell)]^2\, .
$$

\end{proof}

\subsection{Proof of Lemma 2}\label{app_corol_eqD} 

\begin{lemma}
The maps $\mathcal{D}_\ell$ defined by Eqs.~\eref{dissipator}-\eref{jumps} with functions $\Gamma_\ell(\Delta_N^\ell)$ obeying Assumption \ref{Gamma} are such that 
\begin{eqnarray*}
	& \left\|\mathcal{D}_\ell[O]-\Gamma_{\ell}^2(\Delta_N^\ell) \mathcal{D}_\ell^{\rm Loc}[O]\right\|\le \frac{C_O}{N}\, ,\\
	& \left\|\mathcal{D}_\ell[X_{N}]-\Gamma_{\ell}^2(\Delta_N^\ell) \mathcal{D}_\ell^{\rm Loc}[X_N]\right\|\le \frac{C_{x}}{N}\, , \qquad 
\end{eqnarray*}
with 
\begin{equation}
\mathcal{D}^{\rm Loc}_\ell[A]=\frac{1}{2}\sum_{k=1}^N \left(\left[{j}_{\ell}^{ \dagger\, (k)}, A\right] {j}_{\ell}^{(k)} + {j}_{\ell}^{ \dagger\, (k)}[A, {j}_{\ell}^{(k)}]  \right).
\end{equation}
and 
$C_O$, $C_{x}$ appropriate $N$ independent constants. In the above expression, $O$ is any local operator with support on a finite number of sites, $N_O$ is the extension of its support, and $X_N$ an average operator of the single-particle operator $x$ as defined in Eq.~\eref{eq:average-operators}
\end{lemma}

\begin{proof} 
To prove the Lemma, we start considering an operator $O$ with local support, extended over $N_O$ sites, and compute the action of $\mathcal{D}_\ell $ on it. We have
$$
\mathcal{D}_\ell[O]=\frac{1}{2}\sum_{k=1}^N\left(\left[\Gamma_\ell(\Delta_N^\ell) j_\ell^{\dagger\, (k)}, O\right]j_{\ell}^{(k)}\Gamma_\ell(\Delta_N^\ell)+\Gamma_\ell(\Delta_N^\ell)j_\ell^{\dagger \, (k)}\left[O,j_\ell^{(k)}\Gamma_\ell(\Delta_N^\ell)\right]\right)\, .
$$
Using that $[AB,C]=A[B,C]+[A,C]B$, we rewrite this as
\begin{equation}
\eqalign{
\mathcal{D}_\ell[O]=&\frac{1}{2}\sum_{k=1}^N\left(\Gamma_\ell(\Delta_N^\ell)\left[ j_\ell^{\dagger\, (k)}, O\right]j_{\ell}^{(k)}\Gamma_\ell(\Delta_N^\ell)+\Gamma_\ell(\Delta_N^\ell) j_\ell^{\dagger\, (k)}\left[O,j_\ell^{(k)}\right]\Gamma_\ell(\Delta_N^{\ell})\right)\\
+&\frac{1}{2}\sum_{k=1}^N\left(\left[\Gamma_\ell(\Delta_N^\ell), O\right]j_\ell^{\dagger\, (k)}j_{\ell}^{(k)}\Gamma_\ell(\Delta_N^\ell)+\Gamma_\ell(\Delta_N^\ell)j_\ell^{\dagger\, (k)}j_\ell^{(k)}\left[O,\Gamma_\ell(\Delta_N^\ell)\right]\right)\, .
\label{diss-two-terms}
}
\end{equation}
Let us start considering the first term on the right-hand side, which we call $D_1$. Due to the locality of $O$, we can truncate the sum to $k_{min}$ and $k_{max}$ which define the support of $O$. That is, 
$$
D_1=\frac{1}{2}\sum_{k=k_{min}}^{k_{max}}\left(\Gamma_\ell(\Delta_N^\ell)\left[ j_\ell^{\dagger\, (k)}, O\right]j_{\ell}^{(k)}\Gamma_\ell(\Delta_N^\ell)+\Gamma_\ell(\Delta_N^\ell) j_\ell^{\dagger\, (k)}\left[ O, j_{\ell}^{(k)}\right]\Gamma_\ell(\Delta_N^\ell)\right)\, .
$$
Now, we define the following operators 
$$
\tilde{O}_1=\frac{1}{2}\sum_{k=k_{min}}^{k_{max}}\left[ j_\ell^{\dagger\, (k)}, O\right]j_{\ell}^{(k)}\, , \qquad  \tilde{O}_2=\frac{1}{2}\sum_{k=k_{min}}^{k_{max}}j_\ell^{\dagger\, (k)}\left[O , j_{\ell}^{(k)}\right]\, 
$$
which are local, have the same support of $O$, and are such that $\|\tilde{O}_{1/2}\|\le N_O\|O\|\|j_\ell\|^2 $. Through such operators we write
\begin{equation}
\eqalign{
D_1&=\Gamma_\ell(\Delta_N^\ell)\tilde{O}_1\Gamma_\ell(\Delta_N^\ell)+\Gamma_\ell(\Delta_N^\ell)\tilde{O}_2\Gamma_\ell(\Delta_N^\ell)=\\
&=\Gamma_\ell^2(\Delta_N^\ell)\left(\tilde{O}_1+\tilde{O}_2\right)+\Gamma_\ell(\Delta_N^\ell)\left(\left[\tilde{O}_1,\Gamma_\ell(\Delta_N^\ell)\right]+\left[\tilde{O}_2,\Gamma_\ell(\Delta_N^\ell)\right]\right)\, .
}
\label{D_1-two-terms}
\end{equation}
Now, it is important to note that $\tilde{O}_1+\tilde{O}_2=\mathcal{D}_\ell^{\rm Loc}[O]$ and thus that 
$$
D_1=\Gamma_\ell^2(\Delta_N^\ell)\mathcal{D}_\ell^{\rm Loc}[O]+\Gamma_\ell(\Delta_N^\ell)\left(\left[\tilde{O}_1,\Gamma_\ell(\Delta_N^\ell)\right]+\left[\tilde{O}_2,\Gamma_\ell(\Delta_N^\ell)\right]\right)\, .
$$
The first term on the right-hand side, which we call $D_{11}$ is already the term which we expect the quantity $\mathcal{D}_\ell[O]$ to converge to. We thus have to show that the rest, i.e., the second term in $D_1$ and the second term on the right hand side of the Eq.~\eref{diss-two-terms}, is vanishingly small in the large $N$ limit. For what concerns the second term in Eq.~\eref{D_1-two-terms}, which we call $D_{12}$, using relation $i)$ in Lemma \ref{lemma_commutators}, we find
$$
\|D_{12}\|=\left\|\Gamma_\ell(\Delta_N^\ell)\left(\left[\tilde{O}_1,\Gamma_\ell(\Delta_N^\ell)\right]+\left[\tilde{O}_2,\Gamma_\ell(\Delta_N^\ell)\right]\right)\right\|\le 4 \gamma(\delta_\ell)\frac{N_O^2}{N}\delta_\ell \gamma'(\delta_\ell)\|O\|\|j_\ell\|^2\, .
$$
We are thus left with the second term in Eq.~\eref{diss-two-terms}. This is given by 
\begin{equation}
    \eqalign{
D_2&=\frac{1}{2}\sum_{k=1}^{N}\left(\left[\Gamma_\ell(\Delta_N^\ell),O\right]j_\ell^{\dagger \, (k)} j_\ell^{(k)} \Gamma_\ell(\Delta_N^\ell)-\Gamma_\ell(\Delta_N^\ell) j_\ell^{\dagger \, (k)} j_\ell^{(k)}\left[\Gamma_\ell(\Delta_N^\ell),O\right]\right)=\\
&=\frac{1}{2}\sum_{k=1}^{N}\left[\Gamma_\ell(\Delta_N^\ell),O\right]\left[j_\ell^{\dagger \, (k)} j_\ell^{(k)}, \Gamma_\ell(\Delta_N^\ell)\right]+
\\
&+\frac{1}{2}\sum_{k=1}^N
\left(\left[\Gamma_\ell(\Delta_N^\ell),O\right]\Gamma_\ell(\Delta_N^\ell)j_\ell^{\dagger\, (k)}j_\ell^{(k)} -\Gamma_\ell(\Delta_N^\ell)j_\ell^{\dagger\, (k)}j_\ell^{(k)}\left[\Gamma_\ell(\Delta_N^\ell),O\right]\right)\, .
    }
\end{equation}
Looking at the above equation, we split $D_2$ into two parts. We have
\begin{equation}
\eqalign{
    D_{21}&=\frac{1}{2}\sum_{k=1}^{N}\left[\Gamma_\ell(\Delta_N^\ell),O\right]\left[j_\ell^{\dagger \, (k)} j_\ell^{(k)}, \Gamma_\ell(\Delta_N^\ell)\right] \, ,\\
    D_{22}&=\frac{1}{2}\sum_{k=1}^N
\left[\left[\Gamma_\ell(\Delta_N^\ell),O\right],\Gamma_\ell(\Delta_N^\ell)j_\ell^{\dagger\, (k)}j_\ell^{(k)}\right]\, .
}
\end{equation}
Using Lemma \ref{lemma_commutators}, we have that 
$$
\|D_{21}\|\le \sum_{k=1}^N \frac{2N_O}{N^2}\|O\|\delta_\ell^2 [\gamma'(\delta_\ell)]^2\|j_\ell\|^2\le \frac{2N_O}{N}\|O\|\delta_\ell^2 [\gamma'(\delta_\ell)]^2\|j_\ell\|^2\, .
$$
Next we focus on $D_{22}$. We can write it as $$
D_{22}=\frac{1}{2}\sum_{k=1}^N \left[\left[\Gamma_\ell(\Delta_N^\ell),O\right],\Gamma_\ell(\Delta_N^\ell)\right]j_\ell^{\dagger\, (k)}j_\ell^{(k)}+\frac{1}{2}\sum_{k=1}^N \Gamma_\ell(\Delta_N^\ell)\left[\left[\Gamma_\ell(\Delta_N^\ell),O\right],j_\ell^{\dagger\, (k)}j_\ell^{(k)}\right]\, .
$$
Due to Lemma \ref{lemma_commutators}, the first term above, which we call $D_{221}$, is bounded by 
$$
\left\|D_{221}\right\|=\left\|\frac{1}{2}\sum_{k=1}^N \left[\left[\Gamma_\ell(\Delta_N^\ell),O\right],\Gamma_\ell(\Delta_N^\ell)\right]j_\ell^{\dagger\, (k)}j_\ell^{(k)}\right\|\le \frac{2N_O^2}{N}\|O\|\delta_\ell^2 [\gamma'(\delta_\ell)]^2 \|j_\ell\|^2\, .
$$
For the second term of $D_{22}$, which we call $D_{222}$, we use that (see proof of Lemma \ref{lemma_commutators})
$$
\left[\Gamma_\ell(\Delta_N^\ell),O\right]=\sum_{n=0}^\infty c_\ell^n\sum_{i=0}^{n-1}(\Delta_N^\ell)^i O^\ell (\Delta_N^\ell)^{n-1-i} \, ,$$
with
$$
\quad O^\ell =\left[\Delta_N^\ell,O\right]=\frac{1}{N}\sum_{k=k_{min}}^{k_{max}}\left[\sum_{\alpha=1}^{d^2}r_{\ell\alpha} v_\alpha^{(k)}\right] \, .
$$
Because of this, we have that 
$$
\left[\left[\Gamma_\ell(\Delta_N^\ell),O\right],j_\ell^{\dagger \, (k)}j_\ell^{(k)}\right]=\sum_{n=0}^\infty c_\ell^n \sum_{i=0}^{n-1} \left[(\Delta_N^\ell)^i O^\ell (\Delta_N^\ell)^{n-1-i},j_\ell^{\dagger\, (k)}j_\ell^{ (k)}\right]\, , 
$$
and thus 
\begin{equation}
\eqalign{
    D_{222}&=\frac{1}{2}\sum_{k=1}^N\Gamma_\ell(\Delta_N^\ell)\sum_{n=0}^\infty \sum_{i=0}^{n-1} c_\ell^n \left(\left[(\Delta_N^\ell)^i,j_\ell^{\dagger \, (k)}j_\ell^{(k)}\right]O^\ell (\Delta_N^\ell)^{n-1-i}\right)+\\
    &+\frac{1}{2}\sum_{k=1}^N\Gamma_\ell(\Delta_N^\ell)\sum_{n=0}^\infty \sum_{i=0}^{n-1} c_\ell^n \left((\Delta_N^\ell)^i\left[O^\ell,j_\ell^{\dagger \, (k)}j_\ell^{(k)}\right] (\Delta_N^\ell)^{n-1-i}\right)+\\
    &+\frac{1}{2}\sum_{k=1}^N\Gamma_\ell(\Delta_N^\ell)\sum_{n=0}^\infty \sum_{i=0}^{n-1} c_\ell^n \left((\Delta_N^\ell)^i O^\ell \left[(\Delta_N^\ell)^{n-1-i},j_\ell^{\dagger \, (k)}j_\ell^{(k)}\right]\right)\, .
}
    \label{D_222}
\end{equation}
We note that
$$
\| [(\Delta_N^\ell)^n, j_\ell^{\dagger\, (k)} j_\ell^{(k)} ] \| \le \frac{2}{N} n\delta_\ell^{n}\|j_\ell\|^2 \, ,
$$
and that 
$$
\|[O^\ell,j_{\ell}^{\dagger \, (k)}j_\ell^{(k)}]\|\le \frac{4N_O}{N}\delta_\ell \|O\| \|j_\ell\|^2\, ,\qquad \mbox{ if } \quad k\in[k_{min},k_{max}]\, ,
$$
or $\|[O^\ell,j_{\ell}^{\dagger \, (k)}j_\ell^{(k)}]\|=0$ otherwise. 

Diving into three terms,  $D_{222}',D_{222}'',D_{222}'''$, the three terms appearing in Eq.~\eref{D_222}, through the above bounds we find
\begin{equation}
    \eqalign{
        \|D_{222}'\|&\le \frac{N\gamma(\delta_\ell)}{2}\sum_{n=0}^\infty \sum_{i=0}^{n-1}|c_\ell^n|\frac{2}{N} i \delta_\ell^{n-1} \|j_\ell\|^2\frac{2N_O}{N}\|O\|\delta_\ell \\
        &\le  \frac{2N_O\gamma(\delta_\ell)}{N} \|j_\ell\|^2\|O\|\delta_\ell^2\sum_{n=0}^\infty |c_\ell^n|n^2 \delta_\ell^{n-2}\\
        &= \frac{2N_O\gamma(\delta_\ell)}{N}\|j_\ell\|^2\|O\|\delta_\ell^2 \gamma''(\delta_\ell) \, ,
    }
\end{equation}
as well as 
\begin{equation}
    \eqalign{
        \|D_{222}''\|&\le \frac{N_O\gamma(\delta_\ell)}{2}\sum_{n=0}^\infty  |c_\ell^n| \delta_\ell^{n-1}(n-1)\frac{4N_O}{N}\delta_\ell \|O\|\|j_\ell\|^2\\
        &\le \frac{2N_O^2\gamma(\delta_\ell)}{N}\|j_\ell\|^2\delta_\ell \|O\|\gamma'(\delta_\ell)\, ,
    }
\end{equation}
and 
\begin{equation}
    \eqalign{
        \|D_{222}'''\|\le \frac{2N_O\gamma(\delta_\ell)}{N}\|j_\ell\|^2\delta_\ell^2 \|O\|\gamma''(\delta_\ell)\, ,
    }
\end{equation}
just like for $D_{222}'$. 

With all of these bounds, we can now prove the first part of the Lemma. We have
$$
\mathcal{D}_\ell[O]=D_{11}+D_{12}+D_{21}+D_{221}+D_{222}'+D_{222}''+D_{222}'''\, ,
$$
from which we find 
$$
\|\mathcal{D}_\ell[O]-D_{11}\|\le \frac{C_O}{N}\, ,
$$
where $C_O$ is an $N$-independent constant obtained by combining all of the above bounds, and reads
$$
C_O= 2 N_O \|O\| \| j_{\ell} \|^2 \left\lbrace \delta_{\ell}\gamma'(\delta_{\ell})[ \delta_{\ell} \gamma'(\delta_{\ell})(1+N_O)+ 3 N_O \gamma(\delta_{\ell})]  + 2 \gamma(\delta_{\ell}) \delta^2_{\ell} \gamma''(\delta_{\ell}) \right\rbrace.
$$

Now, considering that 
$$
\mathcal{D}_\ell[X_N]=\frac{1}{N}\sum_{k=1}^{N}\mathcal{D}_\ell[x^{(k)}]\, , 
$$
we find 
$$
\left\|\frac{1}{N}\sum_{k=1}^N\left(\mathcal{D}_\ell[x^{(k)}]-\mathcal{D}_\ell^{\rm Loc}[x^{(k)}]\right)\right\|\le \left\|\mathcal{D}_\ell[x^{(k)}]-\mathcal{D}_\ell^{\rm Loc}[x^{(k)}]\right\|\le \frac{C_{x^{(k)}}}{N}\, ,
$$
where $C_{x^{(k)}} $ is an $N$-independent constant reading
$$
C_{x^{(k)}} = 2 \| x \| \| j_{\ell} \|^2 \left\lbrace [2 \delta_{\ell}\gamma'(\delta_{\ell}) +3\gamma(\delta_{\ell})]\delta_{\ell}\gamma'(\delta_{\ell}) + 2 \gamma(\delta_{\ell}) \delta^2_{\ell} \gamma''(\delta_{\ell})\right\rbrace. 
$$

\end{proof}

\subsection{Proof of Lemma 3}\label{app_proof_L3}

\begin{lemma}
Given the generator $\lin_N$ specified by Eqs.~\eref{Lindblad}-\eref{jumps}, with functions $\Gamma_\ell(\Delta_N^\ell)$ obeying Assumption \ref{Gamma}, we have that 
$$
\| \mathcal{L}_N[m^{N}_{\alpha}] -  f_\alpha(\vec{m}^N)\| \leq \frac{C_{L}}{N} \, ,
$$
where 
\begin{eqnarray*}
& f_\alpha(\vec{m}^N) =  i \sum_{\beta = 1}^{d^2} A_{\alpha \beta} m^{N}_{\beta} + i \sum_{\beta,\gamma=1}^{d^2} B_{\alpha \beta \gamma} m^{N}_{\beta} m^{N}_{\gamma} + \sum_{\ell=1}^{q}\sum_{\beta=1}^{d^2} M_{\ell \alpha }^{\beta} \Gamma^{2}_{\ell}(\Delta^{\ell}_N) m^{N}_{\beta} \\
& A_{\alpha \beta} = \sum_{\beta'} \epsilon_{\beta'} a_{\beta' \alpha}^{\beta} \quad B_{\alpha \beta \gamma} = \sum_{\beta'} a_{\beta' \alpha}^{\gamma}(h_{\beta \beta'} + h_{\beta' \beta});
\end{eqnarray*} 
$M$ is a real matrix, such that the action of $\mathcal{D}^{\mathrm{Loc}}_{\ell}[\cdot]$ on an element of the single-site operator basis $v_\alpha^{(k)}$ reads
$$
\mathcal{D}_\ell^{\rm Loc}[v_\alpha^{(k)}]=\sum_{\beta=1}^{d^2}M_{\ell \alpha}^{\beta} v_\beta^{(k)}\, , 
$$
and $C_{L}$ is an $N$ independent constant.
\end{lemma}

\begin{proof}
The proof of this Lemma simply requires the calculation of the action of the Lindblad generator on $m^{N}_{\alpha}$, 
\begin{equation*}
    \lin_N[m_{\alpha}^N]=i[H,m_{\alpha}^N]+\sum_{\ell=1}^{q} \mathcal{D}_\ell[m_{\alpha}^N]\,.
\end{equation*} 
The single-particle Hamiltonian contribution reads
\begin{equation}
\eqalign{
L_1 & =   i\sum_{\beta'}\epsilon_{\beta'}\sum_{k}[v_{\beta'}^{(k)},m_{\alpha}^N]= i\sum_{\beta'}\epsilon_{\beta'}\sum_{k}[v_{\beta'}^{(k)},\frac{1}{N}\sum_{k'} v_{\alpha}^{(k')}] \\ 
&= i \sum_{\beta', \beta} \epsilon_{\beta'} a_{\beta' \alpha}^{\beta} \frac{1}{N}\sum_{k} v_{\beta}^{(k)} \\
& = i \sum_{\beta', \beta} \epsilon_{\beta'} a_{ \beta' \alpha }^{\beta}m^{N}_{\beta} = i \sum_{\beta = 1}^{d^2} A_{\alpha \beta} m^{N}_{\beta}.
}
\end{equation}
The contribution of the all-to-all, two-particle interaction gives instead
\begin{equation*}
\eqalign{
L_2 & = \frac{1}{N} \sum_{k,l}\sum_{\beta, \beta'} h_{\beta \beta'}[v_{\beta}^{(k)}v_{\beta'}^{(l)}, \frac{1}{N}\sum_{i} v_{\alpha}^{(i)}] \\ &= \sum_{\gamma} \sum_{\beta, \beta'} h_{\beta \beta'} (a_{\beta' \alpha}^{\gamma}m^{N}_{\beta}m^{N}_{\gamma} + a_{\beta \alpha}^{\gamma} m^{N}_{\gamma} m^{N}_{\beta'})  \\ 
& = \sum_{\gamma, \beta} B_{\alpha \beta \gamma} m^{N}_{\beta}m^{N}_{\gamma} + \sum_{\gamma} \sum_{\beta, \beta'}h_{\beta' \beta} a_{\beta' \alpha}^{\gamma} [m_{\gamma}^N, m_{\beta}^N].
}
\end{equation*}
We can see that the last term in the second line, that we will denote as $L_{22},$ has a vanishing norm in the thermodynamic limit. Indeed, it is
 $[m_{\gamma}^N, m_{\beta}^N] = \frac{1}{N^2} \sum_k [v_{\gamma}^{(k)}, v_{\beta}^{(k)}]= \frac{1}{N} \sum_{\eta} a_{\gamma \beta}^{\eta} m_{\eta}^{N}$, so that
 \begin{equation}
     \| L_{22}\| \leq \frac{1}{N} d^8 h_{{\mathrm {max}}} a_{max}^2 \, , 
 \end{equation}
where $h_{{\mathrm {max}}} = {\mathrm {max}}_{\beta,\beta'} h_{\beta, \beta'}$ , and $a_{{\mathrm {max}}} = {\mathrm {max}}_{\alpha, \beta, \gamma} a_{\alpha  \beta}^{\gamma}$.
As for the dissipative term, that we denote as $L_3= \sum_{\ell=1}^{q} \mathcal{D}_{\ell}[m^N_{\alpha} ]$, from Lemma \ref{Cor_diss_dyn} it is
\begin{equation*}
\eqalign{
\| L_3 - \sum_{\ell=1}^{q} \Gamma_{\ell}^{2}(\Delta_{N}^{\ell}) \mathcal{D}^{\mathrm{Loc}}_{\ell}[m^{N}_{\alpha} ] \| \leq q\frac{C_v}{N} \, ,
}
\end{equation*}
where $C_v = {\mathrm {max}}_{\forall \ell} \left\lbrace 2 \| j_{\ell} \|^2  [2\delta_{\ell} \gamma'(\delta_{\ell}) +3\gamma(\delta_{\ell})] \delta_{\ell} \gamma'(\delta_{\ell}) +2 \gamma(\delta_{\ell}) \delta^2_{\ell} \gamma''(\delta_{\ell}) \right\rbrace$. By considering the three contribution $L_{1,2,3}$, it is
\begin{equation}
    \|L_1 + L_2+ L_3 - f_{\alpha}(\vec{m}^{N}) \| \leq \frac{1}{N} (d^8 h_{{\mathrm {max}}} a^{2}_{{\mathrm {max}}} +qC_v) \, ,
\end{equation}
from which we find $C_{L} = d^8 h_{{\mathrm {max}}} a^{2}_{{\mathrm {max}}} +qC_v$.
\end{proof} 

\subsection{Proof of Lemma 4}\label{app_proof_L4}

\begin{lemma}
The system of equations \eref{eq:mean-field} with initial conditions $m_\alpha(0)$, defined by a quantum state $\omega$ as in Eq.~\eref{init-cond}, has a unique solution for $t\in[0,\infty)$. Moreover, one has 
$$
|m_\alpha(t)|\le \|v_\alpha\|\le 1\, , \qquad \forall t\, .
$$
\end{lemma}

\begin{proof} We write the system of differential equations appearing in Eq.~\eref{eq:mean-field} in a vector form as 
$$
\frac{d}{dt}\vec{m}=\vec{f}(\vec{m})\, , 
$$
where $\vec{m}=(m_1,m_2,\dots m_{d^2})$ and $\vec{f}(\vec{m})=(f_1(\vec{m}),f_2(\vec{m}),\dots f_{d^2}(\vec{m}))$. The initial condition for the above differential equations is given by $\vec{m}(0)$ which is obtained as the limit 
$$
m_\alpha(0)=\lim_{N\to\infty}\omega( m_\alpha^N)\, .
$$
The fact that the initial condition is obtained from a well-defined quantum state $\omega$ means that we have $|m_\alpha(0)|\le 1$ for all $\alpha=1,2,\dots d^2$. \\

The functions $f_\alpha(\vec{m})$ are made by polynomial terms and by the functions $\Gamma_\alpha$ which are continuous and differentiable by assumption. As such, we have that $\vec{f}$ is continuous and differentiable in the whole $\mathbb{R}^{d^2}$, i.e., $\vec{f}\in C^1(\mathbb{R}^{d^2})$. By the fundamental existence and uniqueness theorem, we can thus conclude that the system of differential equations has a unique solution $\vec{m}(t)$ in the (right) maximal interval of existence $t\in[0,T)$, for $T>0$. 

In order to show that for the above system of equations, $T=\infty$, we need to demonstrate that $\vec{m}(t)$ is contained in a compact set $K\subset \mathbb{R}^{d^2}$. Indeed, whenever $\vec{m}(t)$ belongs to a compact set (i.e., whenever this is bounded), one can conclude that $T=+\infty$. This is the contraposition of the statement that, whenever $T<\infty$, there must exist a time $t\in(0,T)$ such that the solution of the differential equation $\vec{m}\notin K$ with $K$ any compact set in $\mathbb{R}^{d^2}$ (see, e.g., Theorem 3 in Chapter 2 of Ref.~\cite{Perko13}). In few words, we need to show that all the $m_\alpha(t)$ remain bounded. 

To this end, we will compare the time evolution of the variable $m_\alpha$ with the time evolution of the average of the operators $v_\alpha$ that can be obtained through an effective dynamics. Let us thus consider the auxiliary dynamical generator (see also Theorem \ref{theorem2} in the main text)
\begin{equation}
\tilde{\mathcal{L}}_t\left[\cdot \right]= i\left[\tilde{H},\cdot\right]+\sum_{\ell}\Gamma^2_\ell(\Delta^\ell(t))\mathcal{D}_\ell^{\rm Loc}\left[\cdot\right]\, ,
    \label{app:ql-gen}
\end{equation}
where $\mathcal{D}_\ell^{\rm Loc}$ is the dissipator introduced in Eq.~\eref{D_loc}, $\Delta_\ell(t)$ is the linear combination of mean-field variables  
$$
\Delta_\ell(t)=\sum_\beta r_{\ell\beta} m_\beta(t)\, , 
$$
and 
$$
\tilde{H}=\sum_{k=1}^N \sum_{\alpha=1}^{d^2} \epsilon_\alpha v_\alpha^{(k)}+\sum_{k=1}^N \sum_{\alpha,\beta=1}^{d^2} h_{\alpha \beta} \left(m_\alpha(t) v_\beta^{(k)}+m_\beta(t) v_\alpha^{(k)}\right)\, .
$$
Since the functions $m_\alpha(t)$ are well-defined in the interval $[0,T)$, the above generator is also well-defined in such interval. The above generator acts on the different single-particles separately and implements a permutation invariant dynamics. We now calculate the Heisenberg equations of motion for the single-particle observables $v_\alpha^{(k)}$ at a given site $k$. We find that 
$$
\frac{d}{dt}v_\mu^{(k)}=\sum_{\nu=1}^{d^2} G_{\mu \nu} v_\nu^{(k)}\, , 
$$
where we have 
$$
G_{\mu \nu}=i\sum_{\alpha=1}^{d^2} \epsilon_\alpha a_{\alpha \mu}^\nu+i\sum_{\alpha, \beta=1}^{d^2} (h_{\alpha \beta}m_\alpha a_{\beta\mu}^\nu+h_{\alpha\beta}m_\beta a_{\alpha\mu}^\nu)+\sum_{\ell=1}^q \Gamma_\ell^2(\Delta_\ell)M_{\ell \mu }^\nu \, .
$$
In the above equations, we have dropped the time dependence from all operators and mean-field variables for compactness. 
Taking the expectation value of the operators $v_\mu$ with a translation-invariant quantum state $\omega$, we find the following system of differential equations 
$$
\frac{d}{dt}\omega (v_\mu)=\sum_{\nu=1}^{d^2} G_{\mu\nu}\omega (v_\nu)\, ,
$$
and we pick the initial state to be such that $\omega(v_\mu)(0)=m_\mu(0)$. Inspecting the structure of the functions $f_\mu$, it is possible to see that the mean-field equations can actually be recast as 
$$
\frac{d}{dt}m_\mu =\sum_{\mu,\nu=1}^{d^2} G_{\mu\nu} m_\nu \,.
$$
We thus introduce the functions $y_\mu=\omega(v_\mu)-m_\mu$, for which we find the following system of differential equations
$$
\frac{d}{dt}y_\mu =\sum_{\mu,\nu=1}^{d^2} G_{\mu\nu} y_\nu \,.
$$
This is a system of first-order linear differential equations with time-dependent coefficients and thus, since $y_\mu(0)=0$ $\forall \mu$ by construction, we have that $y(t)\equiv 0$. This allows us to conclude that  $m_\mu (t)=\omega(v_\mu)(t)$. Then, we note that $\omega(v_\mu)(t)=\omega (v_\mu(t))$ and since the operator dynamics $v_\mu(t)$ is implemented by a time-dependent contractive map we have that $\|v_\mu(t)\|\le \|v_\mu\|=1$, which in turns implies
$$
|m_\mu(t)|=|\omega(v_\mu)(t)|\le 1 \, .
$$
\end{proof}

\subsection{Proof of Lemma 5}\label{proof_Lemma-aux}

\begin{lemma}

The convergence of the squared operator-valued rates to the same rates computed in their mean-field scalar function is dominated by the convergence of the mean-field operator to the mean-field scalar functions, namely we have that
\begin{eqnarray*}
& |\omega\left(A^\dagger e^{t\mathcal{L}_N}\left[(\Gamma_\ell^2(\Delta_N^\ell)-\Gamma_\ell^2(\Delta_\ell(t)))X\right] B\right)| \\
& \le  C \| X \| \sum_{\alpha=1}^{d^2} |r_{\ell \alpha}| \sqrt{\omega(A^{\dagger} e^{t \linn}[(m_{\alpha}^N-m_{\alpha}(t))^2]A)}\sqrt{\omega(B^{\dagger}B)} \, ,
\end{eqnarray*}
where $C=2\gamma(\delta_{\ell})\gamma'(\delta_\ell)$.
\end{lemma}

\begin{proof}
Let us start considering the difference inside the action of the generator. We have that 
$$
\Gamma_\ell^2(\Delta_N^\ell)-\Gamma_\ell^2(\Delta_\ell(t))=[\Gamma_\ell(\Delta_N^\ell)-\Gamma_\ell(\Delta_\ell(t))][\Gamma_\ell(\Delta_N^\ell)+\Gamma_\ell(\Delta_\ell(t))]\, . 
$$
Exploiting the power series decomposition of the function $\Gamma_\ell$, we can rewrite 
\begin{eqnarray*}
\Gamma_\ell(\Delta_N^\ell)-\Gamma_\ell(\Delta_\ell(t)) & =\sum_{n=0}^\infty c_\ell^n \left[(\Delta_N^\ell)^n-\Delta_\ell^n(t)\right]\\
& =(\Delta_N^\ell-\Delta_\ell(t))\sum_{n=0}^\infty c_\ell^n \sum_{j=0}^{n-1} (\Delta_N^\ell)^{n-j-1}\Delta_\ell^j(t)\, .
\end{eqnarray*}
Expanding for the definition of $\Delta_N^\ell$, and $\Delta_\ell(t)$, we finally have 
$$
\Gamma_\ell(\Delta_N^\ell)-\Gamma_\ell(\Delta_\ell(t))=\sum_{\alpha}r_{{\ell\alpha}}(m_\alpha^N-m_\alpha(t))\sum_{n=0}^\infty c_\ell^n \sum_{j=0}^{n-1} (\Delta_N^\ell)^{n-j-1}\Delta_\ell^j(t)\, .
$$
Combining everything, we find
$$
\Gamma_\ell^2(\Delta_N^\ell)-\Gamma_\ell^2(\Delta_\ell(t))=\sum_{\alpha}r_{{\ell\alpha}}(m_\alpha^N-m_\alpha(t))Q_\ell\, , 
$$
where 
$$
Q_\ell=[\Gamma_\ell(\Delta_N^\ell)+\Gamma_\ell(\Delta_\ell(t))]\sum_{n=0}^\infty c_\ell^n \sum_{j=0}^{n-1} (\Delta_N^\ell)^{n-j-1}\Delta_\ell^j(t)\, .
$$
We can thus write 
\begin{eqnarray*}
I= &  \omega\left(A^\dagger e^{t\mathcal{L}_N}\left[(\Gamma_\ell^2(\Delta_N^\ell)-\Gamma_\ell^2(\Delta_\ell(t)))X\right] B\right) \\
= & \sum_{\alpha}r_{\ell\alpha}\, \omega\left(A^\dagger e^{t\mathcal{L}_N}\left[(m_\alpha^N-m_\alpha(t))Q_\ell X\right]B\right)\, ,
\end{eqnarray*}
and using the generalized Cauchy-Schwarz inequality in Lemma \ref{Lemma-dilation}, and taking the norm bound for $\|X\|$ and $\|Q_\ell\|$ we find 
$$
|I|\le \|X\| \|Q_\ell\|\sum_{\alpha} |r_{\ell\alpha}|\, \sqrt{\omega\left(A^\dagger e^{t\mathcal{L}_N}\left[(m_\alpha^N-m_\alpha(t))^2\right]A\right) }\sqrt{\omega(B^\dagger B)}\, .
$$
Finally, we note that  
$$
\|Q_\ell\|\le 2\gamma(\delta_\ell)\gamma'(\delta_\ell) \, ,
$$
where 
$$
\gamma'(\delta_\ell)=\sum_{n=0}^\infty n|c_\ell^n|\delta_\ell^{n-1} \, .
$$
Note that we use $\delta_\ell$ here since both $\|\Delta_N^\ell\|$ and, because of Lemma ~\ref{lemma_bound_mf}, also $|\Delta_\ell(t)|$ are smaller than or equal to $\delta_\ell$. 
Clearly, both $\gamma(\delta_\ell)$ and $\gamma'(\delta_\ell)$ are finite. 
\end{proof}

\subsection{Lemma 6}
\label{app_lemma6}
\begin{lemma}
\label{Lemma-dilation} Given any completely positive and unital map $\Lambda[\cdot]$ on the quasi-local algebra $\mathcal{A}$ and a state $\omega$, we have that 
$$
 | \omega\left(A^\dagger \Lambda[C^\dagger D]B\right) | \le \sqrt{\omega(A^\dagger \Lambda[C^\dagger C]A)}\sqrt{\omega\left(B^\dagger \Lambda[D^\dagger D]B\right)}
$$
\end{lemma}

\begin{proof}
The proof of the above Lemma (see also the proof in Ref.~\cite{BenattiEtAl18}) exploits the Stinespring dilation theorem. This states that, given any completely positive map, there exists a unitary operator $U$ acting on an enlarged algebra $\mathcal{A}\otimes \mathcal{B}$ and a state $\tau$ acting only on the algebra $\mathcal{B}$ such that 
\begin{equation}
\Lambda[C]=\tau \left(U^\dagger C\otimes {\bf 1} U\right)\, .
\label{dilation}
\end{equation}
Considering this, we can write 
$$
\omega\left(A^\dagger \Lambda[C^\dagger D]B\right)=\omega\otimes \tau\left(A^\dagger \otimes {\bf 1}[U^\dagger (C^\dagger\otimes {\bf 1}) (D\otimes {\bf 1})U]B\otimes {\bf 1}\right)\, , 
$$
and, using the Cauchy-Schwarz inequality, we have 

\begin{equation}
\eqalign{
|\omega\left(A^\dagger \Lambda[C^\dagger D]B\right)|\le&\sqrt{\omega\otimes \tau\left(A^\dagger \otimes {\bf 1}[U^\dagger (C^\dagger C\otimes {\bf 1}) U]A\otimes {\bf 1}\right)} \times \\
&\sqrt{\omega\otimes \tau\left(B^\dagger \otimes {\bf 1}[U^\dagger (D^\dagger D\otimes {\bf 1}) U]B\otimes {\bf 1}\right)}.}
\end{equation}
Finally, recalling the relation in Eq.~\eref{dilation} we can go back to the map $\Lambda$ to obtain 
$$
|\omega\left(A^\dagger \Lambda[C^\dagger D]B\right)|\le\sqrt{\omega\left(A^\dagger \Lambda[C^\dagger C]A\right)} \sqrt{\omega\left(B^\dagger \Lambda[D^\dagger D]B\right)}\, ,
$$
which concludes the proof. 
\end{proof}
 
\bibliography{DM_bib2}        \bibliographystyle{unsrt}

\begin{thebibliography}{10}

\bibitem{BreuerP:2002}
H.~P. Breuer and F.~Petruccione.
\newblock {\em The theory of open quantum systems}.
\newblock Oxford University Press, Great Clarendon Street, 2002.

\bibitem{Lindblad76}
G.~Lindblad.
\newblock On the generators of quantum dynamical semigroups.
\newblock {\em Commun. Math. Phys.}, 48:119--130, 1976.

\bibitem{diehl2008}
S.~Diehl, A.~Micheli, A.~Kantian, B.~Kraus, H.~P. B{\"u}chler, and P.~Zoller.
\newblock Quantum states and phases in driven open quantum systems with cold
  atoms.
\newblock {\em Nat. Phys.}, 4:878--883, 2008.

\bibitem{diehl2010}
S.~Diehl, A.~Tomadin, A.~Micheli, R.~Fazio, and P.~Zoller.
\newblock {Dynamical Phase Transitions and Instabilities in Open Atomic
  Many-Body Systems}.
\newblock {\em Phys. Rev. Lett.}, 105:015702, 2010.

\bibitem{dallatorre2010}
E.~G. Dalla~Torre, E.~Demler, T.~Giamarchi, and E.~Altman.
\newblock Quantum critical states and phase transitions in the presence of
  non-equilibrium noise.
\newblock {\em Nat. Phys.}, 6:806--810, 2010.

\bibitem{Schindler2013}
P.~Schindler, M.~M{\"u}ller, D.~Nigg, J.~T. Barreiro, E.~A. Martinez,
  M.~Hennrich, T.~Monz, S.~Diehl, P.~Zoller, and R.~Blatt.
\newblock Quantum simulation of dynamical maps with trapped ions.
\newblock {\em Nat. Phys.}, 9:361--367, 2013.

\bibitem{tauber2014}
U.~C. T\"auber and S.~Diehl.
\newblock {Perturbative Field-Theoretical Renormalization Group Approach to
  Driven-Dissipative Bose-Einstein Criticality}.
\newblock {\em Phys. Rev. X}, 4:021010, 2014.

\bibitem{marcuzzi2016}
M.~Marcuzzi, M.~Buchhold, S.~Diehl, and I.~Lesanovsky.
\newblock {Absorbing State Phase Transition with Competing Quantum and
  Classical Fluctuations}.
\newblock {\em Phys. Rev. Lett.}, 116:245701, 2016.

\bibitem{minganti2018}
F.~Minganti, A.~Biella, N.~Bartolo, and C.~Ciuti.
\newblock Spectral theory of liouvillians for dissipative phase transitions.
\newblock {\em Phys. Rev. A}, 98:042118, 2018.

\bibitem{iemini2018}
F.~Iemini, A.~Russomanno, J.~Keeling, M.~Schir\`o, M.~Dalmonte, and R.~Fazio.
\newblock {Boundary Time Crystals}.
\newblock {\em Phys. Rev. Lett.}, 121:035301, 2018.

\bibitem{carollo2019}
F.~Carollo, E.~Gillman, H.~Weimer, and I.~Lesanovsky.
\newblock {Critical Behavior of the Quantum Contact Process in One Dimension}.
\newblock {\em Phys. Rev. Lett.}, 123:100604, 2019.

\bibitem{chertkov2022}
E.~Chertkov, Z.~Cheng, A.~C. Potter, S.~Gopalakrishnan, T.~M. Gatterman, J.~A.
  Gerber, K.~Gilmore, D.~Gresh, A.~Hall, A.~Hankin, M.~Matheny, T.~Mengle,
  D.~Hayes, B.~Neyenhuis, R.~Stutz, and M.~Foss-Feig.
\newblock Characterizing a non-equilibrium phase transition on a quantum
  computer.
\newblock {\em arXiv:2209.12889}, 2022.

\bibitem{sieberer2013}
L.~M. Sieberer, S.~D. Huber, E.~Altman, and S.~Diehl.
\newblock {Dynamical Critical Phenomena in Driven-Dissipative Systems}.
\newblock {\em Phys. Rev. Lett.}, 110:195301, 2013.

\bibitem{helmrich2020}
S.~Helmrich, A.~Arias, G.~Lochead, T.~M. Wintermantel, M.~Buchhold, S.~Diehl,
  and S.~Whitlock.
\newblock Signatures of self-organized criticality in an ultracold atomic gas.
\newblock {\em Nature}, 577:481--486, 2020.

\bibitem{jo2021}
M.~Jo, J.~Lee, K.~Choi, and B.~Kahng.
\newblock Absorbing phase transition with a continuously varying exponent in a
  quantum contact process: A neural network approach.
\newblock {\em Phys. Rev. Research}, 3:013238, 2021.

\bibitem{jo2022}
M.~Jo and M.~Kim.
\newblock Simulating open quantum many-body systems using optimised circuits in
  digital quantum simulation.
\newblock {\em arXiv:2203.14295}, 2022.

\bibitem{hinrichsen2000}
H.~Hinrichsen.
\newblock Non-equilibrium critical phenomena and phase transitions into
  absorbing states.
\newblock {\em Adv. Phys.}, 49:815--958, 2000.

\bibitem{weimer2021}
H.~Weimer, A.~Kshetrimayum, and R.~Or\'us.
\newblock Simulation methods for open quantum many-body systems.
\newblock {\em Rev. Mod. Phys.}, 93:015008, 2021.

\bibitem{BenedikterPS15}
N.~Benedikter, M.~Porta, and B.~Schlein.
\newblock {\em {Effective Evolution Equations from Quantum Dynamics}}.
\newblock SpringerBriefs in Mathematical Physics. Springer International
  Publishing, 2015.

\bibitem{MerkliR18}
M.~Merkli and A.~Rafiyi.
\newblock Mean field dynamics of some open quantum systems.
\newblock {\em Proc. R. Soc. A: Math. Phys. Eng. Sci.}, 474:20170856, 2018.

\bibitem{Porta16}
M.~Porta.
\newblock Mean field dynamics of interacting fermionic systems.
\newblock {\em Mathematical Problems in Quantum Physics}, 717:13, 2016.

\bibitem{Pickl11}
P.~Pickl.
\newblock A simple derivation of mean field limits for quantum systems.
\newblock {\em Lett. Math. Phys.}, 97:151--164, 2011.

\bibitem{hepp1973}
K.~Hepp and E.~H. Lieb.
\newblock {On the superradiant phase transition for molecules in a quantized
  radiation field: the Dicke maser model}.
\newblock {\em Ann. Phys.}, 76:360--404, 1973.

\bibitem{hioe1973}
F.~T. Hioe.
\newblock {Phase Transitions in Some Generalized Dicke Models of
  Superradiance}.
\newblock {\em Phys. Rev. A}, 8:1440--1445, 1973.

\bibitem{alicki1983}
R.~Alicki and J.~Messer.
\newblock Nonlinear quantum dynamical semigroups for many-body open systems.
\newblock {\em J. Stat. Phys.}, 32:299--312, 1983.

\bibitem{Benatti2016}
F.~Benatti, F.~Carollo, R.~Floreanini, and H.~Narnhofer.
\newblock Non-markovian mesoscopic dissipative dynamics of open quantum spin
  chains.
\newblock {\em Phys. Lett. A}, 380:381--389, 2016.

\bibitem{BenattiEtAl18}
F.~Benatti, F.~Carollo, R.~Floreanini, and H.~Narnhofer.
\newblock Quantum spin chain dissipative mean-field dynamics.
\newblock {\em J. Phys. A: Math. Theor.}, 51:325001, 2018.

\bibitem{davies1973}
E.~B. Davies.
\newblock {Exact dynamics of an infinite-atom Dicke maser model}.
\newblock {\em Commun. Math. Phys.}, 33:187--205, 1973.

\bibitem{mori2013}
T.~Mori.
\newblock Exactness of the mean-field dynamics in optical cavity systems.
\newblock {\em J. Stat. Mech. Theory Exp.}, 2013:P06005, 2013.

\bibitem{CarolloL:PRL:21}
F.~Carollo and I.~Lesanovsky.
\newblock {Exactness of Mean-Field Equations for Open Dicke Models with an
  Application to Pattern Retrieval Dynamics}.
\newblock {\em Phys. Rev. Lett.}, 126:230601, 2021.

\bibitem{kirton2017}
P.~Kirton and J.~Keeling.
\newblock {Suppressing and Restoring the Dicke Superradiance Transition by
  Dephasing and Decay}.
\newblock {\em Phys. Rev. Lett.}, 118:123602, 2017.

\bibitem{shammah2018}
N.~Shammah, S.~Ahmed, N.~Lambert, S.~De~Liberato, and F.~Nori.
\newblock Open quantum systems with local and collective incoherent processes:
  Efficient numerical simulations using permutational invariance.
\newblock {\em Phys. Rev. A}, 98:063815, 2018.

\bibitem{huybrechts2020}
D.~Huybrechts, F.~Minganti, F.~Nori, M.~Wouters, and N.~Shammah.
\newblock Validity of mean-field theory in a dissipative critical system:
  Liouvillian gap, $\mathbb{PT}$-symmetric antigap, and permutational symmetry
  in the $\mathit{XYZ}$ model.
\newblock {\em Phys. Rev. B}, 101:214302, 2020.

\bibitem{wang2021}
P.~Wang and R.~Fazio.
\newblock {Dissipative phase transitions in the fully connected Ising model
  with $p$-spin interaction}.
\newblock {\em Phys. Rev. A}, 103:013306, 2021.

\bibitem{piccitto2021}
G.~Piccitto, M.~Wauters, F.~Nori, and N.~Shammah.
\newblock Symmetries and conserved quantities of boundary time crystals in
  generalized spin models.
\newblock {\em Phys. Rev. B}, 104:014307, 2021.

\bibitem{Hopfield:1982}
J.~J. Hopfield.
\newblock Neural networks and physical systems with emergent collective
  computational abilities.
\newblock {\em PNAS}, 79:2554--2558, 1982.

\bibitem{Gayrard92}
V.~{Gayrard}.
\newblock {Thermodynamic limit of the q-state Potts-Hopfield model with
  infinitely many patterns}.
\newblock {\em J. Stat. Phys.}, 68:977--1011, 1992.

\bibitem{MarshEtAl:PhysRevX:21}
B.~P. Marsh, Y.~Guo, R.~M. Kroeze, S.~Gopalakrishnan, S.~Ganguli, J.~Keeling,
  and B.~L. Lev.
\newblock {Enhancing Associative Memory Recall and Storage Capacity Using
  Confocal Cavity QED}.
\newblock {\em Phys. Rev. X}, 11:021048, 2021.

\bibitem{RotondoEtal:2018}
P.~Rotondo, M.~Marcuzzi, J.~P. Garrahan, I.~Lesanovsky, and M.~M\"uller.
\newblock {Open quantum generalisation of Hopfield neural networks}.
\newblock {\em J. Phys. A: Math. Theor.}, 51:115301, 2018.

\bibitem{Fiorelli:PRA:2019}
E.~Fiorelli, P.~Rotondo, M.~Marcuzzi, J.~P. Garrahan, and I.~Lesanovsky.
\newblock {Quantum accelerated approach to the thermal state of classical
  all-to-all connected spin systems with applications to pattern retrieval in
  the Hopfield neural network}.
\newblock {\em Phys. Rev. A}, 99:032126, 2019.

\bibitem{FiorelliLM22}
E.~Fiorelli, I.~Lesanovsky, and M.~M\"uller.
\newblock {Phase diagram of quantum generalized Potts-Hopfield neural
  networks}.
\newblock {\em New J. Phys.}, 24:033012, 2022.

\bibitem{glauber1963}
R.~J. Glauber.
\newblock {Time‐Dependent Statistics of the Ising Model}.
\newblock {\em J. Math. Phys.}, 4:294--307, 1963.

\bibitem{walter2015}
J.-C. Walter and G.T. Barkema.
\newblock {An introduction to Monte Carlo methods}.
\newblock {\em Phys. A: Stat. Mech. Appl.}, 418:78--87, 2015.
\newblock Proceedings of the 13th International Summer School on Fundamental
  Problems in Statistical Physics.

\bibitem{fredrickson1984}
G.~H. Fredrickson and H.~C. Andersen.
\newblock {Kinetic Ising Model of the Glass Transition}.
\newblock {\em Phys. Rev. Lett.}, 53:1244--1247, 1984.

\bibitem{cancrini2008}
N.~Cancrini, F.~Martinelli, C.~Roberto, and C.~Toninelli.
\newblock Kinetically constrained spin models.
\newblock {\em Probab. Theory. Relat. Fields.}, 140:459--504, 2008.

\bibitem{garrahan2011}
J.~P Garrahan, P.~Sollich, and C.~Toninelli.
\newblock Kinetically constrained models.
\newblock {\em in "Dynamical heterogeneities in glasses, colloids, and granular
  media", Eds.: L. Berthier, G. Biroli, J.-P. Bouchaud, L. Cipelletti and W.
  van Saarloos (Oxford University Press, 2011)}, 150:111--137, 2011.

\bibitem{Garrahan18}
J.~P. Garrahan.
\newblock Aspects of non-equilibrium in classical and quantum systems: Slow
  relaxation and glasses, dynamical large deviations, quantum non-ergodicity,
  and open quantum dynamics.
\newblock {\em Physica A: Statistical Mechanics and its Applications},
  504:130--154, 2018.
\newblock Lecture Notes of the 14th International Summer School on Fundamental
  Problems in Statistical Physics.

\bibitem{CarolloGK:JSP:21}
F.~Carollo, J.~P. Garrahan, and R.~L. Jack.
\newblock {Large Deviations at Level 2.5 for Markovian Open Quantum Systems:
  Quantum Jumps and Quantum State Diffusion}.
\newblock {\em J. Stat. Phys.}, 184, 2021.

\bibitem{BratteliR82}
O.~Bratteli and D.~W. Robinson.
\newblock {\em Operator Algebras and Quantum Statistical Mechanics II.
  Equilibrium States Models in Quantum Statistical Mechanics.}
\newblock Springer Berlin, Heidelberg, 1981.

\bibitem{Verbeure10}
A.~F. Verbeure.
\newblock {\em Many-Body Boson Systems: Half a Century Later}.
\newblock Springer London, 2011.

\bibitem{LandfordR69}
O.~E.~Lanford III and D.~Ruelle.
\newblock {Observables at infinity and states with short range correlations in
  statistical mechanics}.
\newblock {\em Commun. Math. Phys.}, 13:194 -- 215, 1969.

\bibitem{Strocchi05}
F.~Strocchi.
\newblock {\em Symmetry breaking}.
\newblock Springer Berlin, Heidelberg, 2021.

\bibitem{thirring2013quantum}
W.~Thirring.
\newblock {\em Quantum mathematical physics: atoms, molecules and large
  systems}.
\newblock Springer Science \& Business Media, 2013.

\bibitem{grimmett2020probability}
G.~Grimmett and D.~Stirzaker.
\newblock {\em Probability and random processes}.
\newblock Oxford university press, 2020.

\bibitem{Amit_book}
Daniel~J. Amit.
\newblock {\em Modelling Brain Function: The World of Attractor Neural
  Networks}.
\newblock Cambridge University Press, USA, 1st edition, 1992.

\bibitem{AmitGS:1985a}
Daniel~J. Amit, H.~Gutfreund, and H.~Sompolinsky.
\newblock Spin-glass models of neural networks.
\newblock {\em Phys. Rev. A}, 32:1007--1018, 1985.

\bibitem{Rotondo:JPA:2018}
P.~Rotondo, M.~Marcuzzi, J.~P. Garrahan, I.~Lesanovsky, and M.~M\"uller.
\newblock {Open quantum generalisation of Hopfield neural networks}.
\newblock {\em J. Phys. A: Math. Theor.}, 51:115301, 2018.

\bibitem{williamson2020}
L.~A. Williamson, M.~O. Borgh, and J.~Ruostekoski.
\newblock {Superatom Picture of Collective Nonclassical Light Emission and
  Dipole Blockade in Atom Arrays}.
\newblock {\em Phys. Rev. Lett.}, 125:073602, 2020.

\bibitem{Perko13}
L.~Perko.
\newblock {\em Differential equations and dynamical systems}.
\newblock Springer Science \& Business Media, 2013.

\end{thebibliography}

\end{document}